\documentclass[journal,draftclsnofoot,onecolumn,12pt]{IEEEtran}
\usepackage{amsthm,amssymb ,graphicx,multirow,amsmath,color,amsfonts}
\usepackage[update,prepend]{epstopdf}%
\usepackage[noadjust]{cite}
\usepackage{tikz}
\usetikzlibrary{arrows,calc}		
\usepackage{bbm} 
\usepackage{pdfpages}
\usepackage{tabulary}
\usepackage{multirow}
\usepackage{comment}
\usepackage{dsfont}
\usepackage{pgf,tikz}
\usepackage{paralist}
\usepackage{amsfonts}
\usepackage{setspace,lipsum}
\usepackage{url}
\usepackage{soul}
\usepackage{cancel}
\usepackage{bigints}
\usepackage{mathtools}
\usepackage[english]{babel}
\usepackage[autostyle]{csquotes}
\usepackage{textcomp}
\usepackage{varwidth}
\usepackage{tcolorbox}
\usepackage{lipsum}
\usepackage{xcolor}

\let\OLDthebibliography\thebibliography
\renewcommand\thebibliography[1]{
  \OLDthebibliography{#1}
  \setlength{\parskip}{0pt}
  \setlength{\itemsep}{0pt plus 0.3ex}
}

\newcommand{\pgfl}{\mathtt{PGFL}}
\newcommand{\cdf}{\mathtt{CDF}}

\newcommand{\pmf}{\mathtt{pmf}}
\newcommand{\mgf}{\mathtt{mgf}}

\newcommand\numberthis{\addtocounter{equation}{1}\tag{\theequation}}

\newcommand{\x}{\mathbf{x}}

\newcommand{\z}{\mathbf{z}}

\def\nb0{{\mathbf{0}}}
\def\nb1{{\mathbf{1}}}

\def\ncalB{{\mathcal{B}}}

\def\ncalM{{\mathcal{M}}}

\def\ncalZ{{\mathcal{Z}}}

\newtheorem{lemma}{Lemma}

\newtheorem{definition}{Definition}

\newtheorem{theorem}{Theorem}

\newtheorem{cor}{Corollary}

\newtheorem{remark}{Remark}

\def\E{\mathbb{E}}

\def\P{\mathbb{P}}

\def\p{p}

\def\Z{\mathbb{Z}}

\def\sir{\mathtt{SIR}}

\def\cdf{\mathtt{CDF}}

\def\pmf{\mathtt{pmf}}

\allowdisplaybreaks 

\begin{document}
\graphicspath{{./Figures/}}
\title{Spatial Distribution of the Mean Peak Age of Information in Wireless Networks}
\author{
Praful D. Mankar, Mohamed A. Abd-Elmagid, and Harpreet S. Dhillon
\thanks{P. D. Mankar is with SPCRC, IIIT Hyderabad, India (Email: praful.mankar@iiit.ac.in). M. A. Abd-Elmagid and H. S. Dhillon are with Wireless@VT, Department of ECE, Virginia Tech, Blacksburg, VA (Email: \{maelaziz,\ hdhillon\}@vt.edu). The support of the U.S. NSF (Grant CPS-1739642) is gratefully acknowledged. {This paper is submitted in part to IEEE ICC 2021 \cite{Praful_ICC2}.} 
}
}
\maketitle
\begin{abstract}
This paper considers a large-scale wireless network consisting of source-destination (SD) pairs, where the sources send time-sensitive information, termed {\em status updates}, to their corresponding destinations in a time-slotted fashion. We employ Age of information (AoI) for quantifying the freshness of the status updates  measured at the destination nodes  { under the preemptive and non-preemptive queueing disciplines with no storage facility. The non-preemptive queue  drops the newly arriving updates  until the update in service is successfully delivered, whereas the preemptive queue replaces the current update in service with the newly arriving update, if any. }  As the update delivery rate for a given link is a function of the interference field seen from the receiver, the temporal mean AoI can be treated as a random variable over space. Our goal in this paper is to characterize the spatial distribution of the mean AoI observed by the SD pairs by modeling them as a Poisson bipolar process.  Towards this objective, we first derive accurate bounds on the moments of success probability while efficiently capturing the interference-induced coupling  in the activities of the SD pairs.  Using this result, we then derive tight bounds on the moments as well as the spatial distribution of peak AoI (PAoI). Our numerical results verify our analytical findings and demonstrate the impact of various system design parameters on the mean PAoI.
\end{abstract}

\IEEEpeerreviewmaketitle
\section{Introduction}
With the emergence of Internet of Things (IoT), wireless networks are expected to provide a reliable platform for enabling real-time monitoring and control applications. Many such applications, such as the ones related to air pollution or soil moisture monitoring, involve a large-scale deployment of IoT sensors, which would acquire updates about some underlying random process and send them to the destination nodes (or monitoring stations). Naturally, accurate quantification of the freshness of status updates received at the destination nodes is essential in such applications. However, the traditional performance metrics of communication systems, like throughput and delay, are not suitable for this purpose since they do not account for the generation times of status updates. This has recently motivated the use of AoI to quantify the performance of
communication systems dealing with the transmission of time-sensitive information \cite{abd2018role}. This metric was first conceived in \cite{kaul2012real} for a simple queuing-theoretic model in which randomly generated update packets arrive at a source node according to a Poisson process, and then transmitted to a destination node using FCFS queuing discipline. In particular, AoI was defined in \cite{kaul2012real} as the time elapsed since the latest successfully received update packet at the destination node was generated at the source node. As evident from the definition, AoI is capable of quantifying how {\em fresh} the status updates are when they reach the destination node since it tracks the generation time  of each update packet.
As will be discussed next in detail, the analysis of AoI has mostly been limited to simple settings that ignore essential aspects of wireless networks, including temporal channel variations and random spatial distribution of nodes. This paper presents a novel spatio-temporal analysis of AoI in a wireless network by incorporating the effect of both the channel variations and the randomness in the wireless node locations to derive the spatial distribution of the temporal mean AoI.  

\subsection{Prior Art}
\label{sec:prior_art}
 For a point-to-point communication system, the authors of \cite{kaul2012real} characterized a closed-form expression for the average AoI. 
Subsequently,  the authors of \cite{yates2012real,costa2016age,kam2013age,Modiano2015,chen2016age,Basel,kosta2017age,javani2019age}  characterized the average AoI or similar age-related metrics (e.g., PAoI \cite{costa2016age,chen2016age,Basel} and Value of Information of Update \cite{kosta2017age}) for a variety of the queuing disciplines. In addition, the authors of  \cite{8820073,kosta2020non,8737474,9018244,yates2018age} presented the queuing theory-based analysis of  the distribution of AoI. The above works provide foundational understanding of temporal variations of AoI from the perspective of queuing theory for a point-to-point communication system. Inspired by these works, the AoI or similar age-related metrics have been used to characterize the performance  of real-time monitoring services in a variety of communication systems, including  broadcast networks \cite{kadota2016,hsu2019scheduling,bastopcu2020should}, multicast networks \cite{Buyukates_ulu,li2020age_a}, multi-hop networks \cite{talak2017,2_3}, multi-server information-update systems \cite{ABedewy2016}, IoT networks \cite{gu2019timely,abd2019tcom,zhou2018joint,AbdElmagid2019Globecom_a,Stamatakis_2020,AbdElmagid_joint}, cooperative device-to-device (D2D) communication networks \cite{buyukates2019age,li2020age_b}, unmanned aerial vehicle (UAV)-assisted networks \cite{abd2018average,AbdElmagid2019Globecom_b}, ultra-reliable low-latency vehicular networks \cite{abdel2018ultra}, and social networks \cite{bastopcu2019minimizing,altman2019forever,hargreaves2019often}.  All these studies mostly focus on minimizing the AoI with the following design objectives:  1) design of scheduling policies  \cite{kadota2016,hsu2019scheduling,bastopcu2020should,talak2017,AbdElmagid_joint}, 2) design of cooperative transmission policies  \cite{Buyukates_ulu,li2020age_a,talak2017,2_3,buyukates2019age,li2020age_b}, 3) design of the status update sampling policies  \cite{zhou2018joint,Stamatakis_2020,AbdElmagid2019Globecom_a,bastopcu2019minimizing}, and 4)  trade-off with other performances metrics in heterogeneous traffic/networks scenarios \cite{2_3,ABedewy2016,gu2019timely,Stamatakis_2020,altman2019forever}.  However, given the underlying tools used,   these works are not conducive to account for some important aspects of wireless networks, such as interference, channel variations, path-loss, and random network topologies.

{ Over the last decade, the stochastic geometry has emerged as a powerful tool for the analysis of the large scale random wireless networks while efficiently capturing the above propagation features. The interested readers, for examples, can refer to models and analyses presented  in  \cite{AndBacJ2011} for cellular networks,  \cite{DhiGanJ2012} for heterogeneous networks, and \cite{Baccelli_Aloha2006} for ad-hoc networks. }
{While these works were initially focused on the space-time mean performance of network  (such as {\em coverage probability}), recently a new performance metric is recently introduced, termed {\em meta distribution}, in \cite{Martin2016Meta} for characterizing the spatial variation in the temporal mean  performance measured at the randomly located nodes. The meta distribution has become an instrumental tool for analysing the spatial disparity in a variety of performance metrics (e.g., see \cite{Kalamkar_Rate_Reliability_Two_Facets,Zhong_SpatoiTemporal,emara2019spatiotemporal})  and network settings (e.g., see \cite{Chiranjib_MetaDis_PCP,Praful_NOMA}).
However, these stochastic geometry models lack to handle the traffic variations because of which they are mostly applicable to saturated networks. Therefore, it is important to develop a method/tool that is capable of handling in spatial randomness through interference and temporal variations due to traffic. But, in general,  the spatiotemporal analysis is known to be hard, for the reasons discussed next.}

{ The wireless links exhibit the space-time correlation through their interference-induced stochastic interactions \cite{Ganti_2009}. This implies that the service rates of the wireless nodes are also correlated (in both space and time), which, as a result, generates coupling between the activities of their associated queues under random traffic patterns.
The coupled queues cause difficulty in the spatiotemporal analysis of wireless networks. 
It is worth noting that the exact characterization of the correlated queues is unknown even in  simple settings  (refer to \cite{Rao_1988}). In the existing literature, there are  broadly two approaches for the spatiotemproal analyses trying to capture the temporal traffic dynamics and spatial nodes variation to some extend. The first approach is focused on the development of iterative frameworks/algorithms specific to the performance metric of interest. 
In this approach, the queues are first  decoupled by modeling the activity of each node independently using  a spatial mean activity \cite{Bartllomiej2016,Emara_2020,Chisci_2019,Elsawy_2020,Howard_Yang_2020}  (or,   a spatial distribution \cite{Howard_Yang_2019,Howard_Yang_2020a,yang2020age,emara2019spatiotemporal}) and then solve the analytical framework in an iterative manner until the spatial mean (or,  the spatial distribution) converge to its {\em fixed-point solution}. On the other hand, the second approach adopts to the  approach used in queueing theory literature for obtaining the performance bounds on correlated queues (for example, see \cite{Rao_1988}). In this approach, the activities of nodes are decoupled by constructing  {\em a dominant system}  wherein the interfering nodes are considered to have the saturated queues \cite{bonald2004wireless,Haenggi_Stability,Zhong_SpatoiTemporal,hu2018age}. 
Because of their saturated queues, the  activities of  interfering nodes will be increased. Thus,  the observing node will overestimate the interference power, and hence, needless to say, its performance will be a bound.
However, such a bound tends to get loosen when the  traffic conditions are lighter, which is quite intuitive. In  \cite{bonald2004wireless},  a second degree of dominant system is presented wherein the interferers are assumed to operate under their corresponding dominant systems. Naturally, this second degree modifications will provide a better performance bound as it estimates the interference more accurately  compared to  the dominant  system. In contrast, the coupling between their activities become insignificant when a massive number of nodes access the channel in a sporadic manner and thus it can be  safely  ignored in the analysis \cite{Y_Zhong}.
} 

{ Nevertheless, there are only a handful of recent works focusing on the spatiotemproal analysis of AoI for wireless networks.  For example, for a Poisson bipolar network,  \cite{yang2020age} derived the spatiotemporal average of AoI for an infinite-length first-come-first-served (FCFS) queue, whereas \cite{yang2020optimizing} derived the spatiotemporal mean PAoI for a unit-length last-come-first-served (LCFS) queue with replacement. The authors first obtained a fixed point solution to the meta distribution in an iterative manner and then applied it to determine the spatiotemporal mean AoI. The authors of \cite{yang2020optimizing} also investigated a locally adaptive scheduling policy that minimizes the mean AoI.
Further, \cite{hu2018age} derived the upper and lower bounds on the cumulative distribution function ($\cdf$) of the temporal mean AoI for a Poisson bipolar network using the construction of dominant system.
On the other hand, the spatiotemporal analysis peak AoI for uplink Internet-of-things (IoT) networks is presented in \cite{emara2019spatiotemporal,Emara_2020,Praful_GC1} by modeling the locations of BSs and IoT devices using independent Poisson point process (PPP). The authors of \cite{emara2019spatiotemporal} derived the mean peak AoI for time-triggered (TT) and event-triggered (ET) traffic. The authors employed an iterative framework wherein quantized meta distribution  and spatial average activities of devices with different classes (properly constructed based on TT and ET traffic) are determined together.  Further, \cite{Emara_2020} derived the mean peak AoI for the prioritized multi-stream traffic (PMT) by iteratively solving the queueing theoretic framework (developed for PMT) and priority class-wise successful transmission probabilities together. 
The authors of \cite{Praful_GC1} derived the  spatial distribution of the temporal mean peak AoI while assuming the IoT devices sample their updates using {\em generate-at-will} policy (see \cite{abd2018role}).  
 Besides, it is worth noting that the delay analyses presented in \cite{Chisci_2019,Elsawy_2020,Howard_Yang_2019,Howard_Yang_2020a} can be extended to analyze the PAoI under FCFS queuing discipline. 
}

\subsection{Contributions}
\label{subsec:Contribution}
This paper presents a stochastic geometry-based analysis of AoI for a large-scale wireless network, wherein the sources transmit time-sensitive status updates to  their corresponding destinations. In particular,  we derive the spatial distribution of  mean PAoI  while assuming that the locations of SD pairs follow a homogeneous bipolar PPP.
In order to overcome the challenge of interference-induced coupling across queues associated with different SD pairs, we propose a tractable two-step analytical approach which relies on a careful construction of {\em dominant systems}. 
The proposed framework efficiently captures the stochastic interaction in both space (through interference) and time (through random transmission activities).  Our approach provides a much tighter lower bound on the spatial moments of the conditional (location-dependent) successful transmission probability compared to the exiting stochastic geometry-based analyses, e.g.,  \cite{Zhong_SpatoiTemporal}, which mainly rely on the assumption of having saturated queues at the interfering nodes.

 { The above construction of dominant system allows to model the upper bound of service times of the update transmissions using geometric distribution.  Further, assuming the Bernoulli arrival of status update, we model status update transmissions using $\mathtt{Geo}/\mathtt{Geo}/1$ queue with no storage facility under preemptive and non-preemptive disciplines. For this setup, we  present the spatiotemproal analysis of the PAoI.
In particular, we derive tight upper bounds on the spatial moments of the temporal mean PAoI for the above queue disciplines.
} 
The contributions of this paper are briefly summarized as below.  
\begin{enumerate}
\item This paper presents a novel analytical framework to determine close lower bounds on the moments of the conditional success probability  while efficiently capturing the interference-induced coupling in the activities of SD pairs.
\item We derive the temporal mean  of PAoI under  {both preemptive and non-preemptive} queueing disciplines for a given success  probability of transmission.
\item Next, using the lower bounds on the moments of the conditional success probability, we derive tight upper bounds on the spatial moments of the temporal mean PAoI for both queuing disciplines.    
\item Using the beta approximation for the distribution of conditional success probability \cite{Martin2016Meta}, we also  characterize the spatial distribution of temporal mean PAoI.
\item Next, we validate the accuracy of the proposed analytical framework for AoI analysis through extensive simulations. Finally, our numerical results reveal the impact of key design parameters, such as the medium access probability, update arrival rate, and signal-to-interference ($\sir$) threshold, on the spatial mean and standard deviation of the temporal mean PAoI observed under the aforementioned  queuing disciplines.  
\end{enumerate}
\section{System Model}
\label{Sec:System_Model}
{ We model SD pairs using a static network wherein the locations of sources are distributed according to a PPP $\Phi$ with density $\lambda_\text{sd}$  and their corresponding destinations are located at fixed distance $R$ from them in uniformly random  directions. 
The $\sir$ measured at the  destination at $\z$ in the $k$-th transmission slot is  
\begin{equation}
\sir_{\z,k}=\frac{h_{\z}^kR_o^{-\alpha}}{\sum_{\x\in\Phi}h_{\x,\z}^k\|\x-\z\|^{-\alpha}\mathds{1}(\x\in\Phi_k)},
\label{eq:SIR}
\end{equation}
where $\Phi_k$ is the set of sources with active transmission during the $k$-th transmission slot, and {$\mathds{1}(\x\in\Phi_k)$ is 1 if $\x\in\Phi_k$, otherwise 0.}  $\alpha$ is the path-loss exponent, and $h_\z^k$ and $h_{\x,\z}^k$ are the channel gains for the link from deriving source and the link from the source at $\x$ to the destination at $\z$, respectively, in transmission slot $k$.   We assume quasi-static Rayleigh fading and model, which implies  $h_{\x,\z}^k\sim{\rm exp}(1)$  independently across both $\x\in\Phi$ and $k\in\mathbb{N}$. 

Since the point process of SD pairs is assumed to be stationary, the distribution of  $\sir$ observed by  the different devices across a large scale static realization of network  is equivalent to the  distribution of $\sir$ measured at a fixed point across multiple realizations. This can be formalize using {\em Palm distribution}, which is the conditional distribution of the point process given that the typical point is present  at a fixed location. Further, by the virtue of Slivnyak's theorem, we know that the Palm distribution for PPP is same as the original distribution of PPP. Please refer to \cite{haenggi2012stochastic} for more details.
 Therefore, we place the  typical SD pair link such that its source and destination are at the origin $o$ and $\x_o\equiv [R,0]$, respectively. Fig. \ref{Fig:Illustration} presents a representative realization of the Poisson bipolar network with the typical link at $(o,\x_o)$. As the analysis is focused on this typical link,  we drop the subscript $\z$ from $\sir_{\z,k}$ and $h_{\x,\z}^k$ here onwards.}
  \begin{figure}[h]
\centering
 \includegraphics[trim=0cm .8cm 0cm .8cm, width=.55\textwidth]{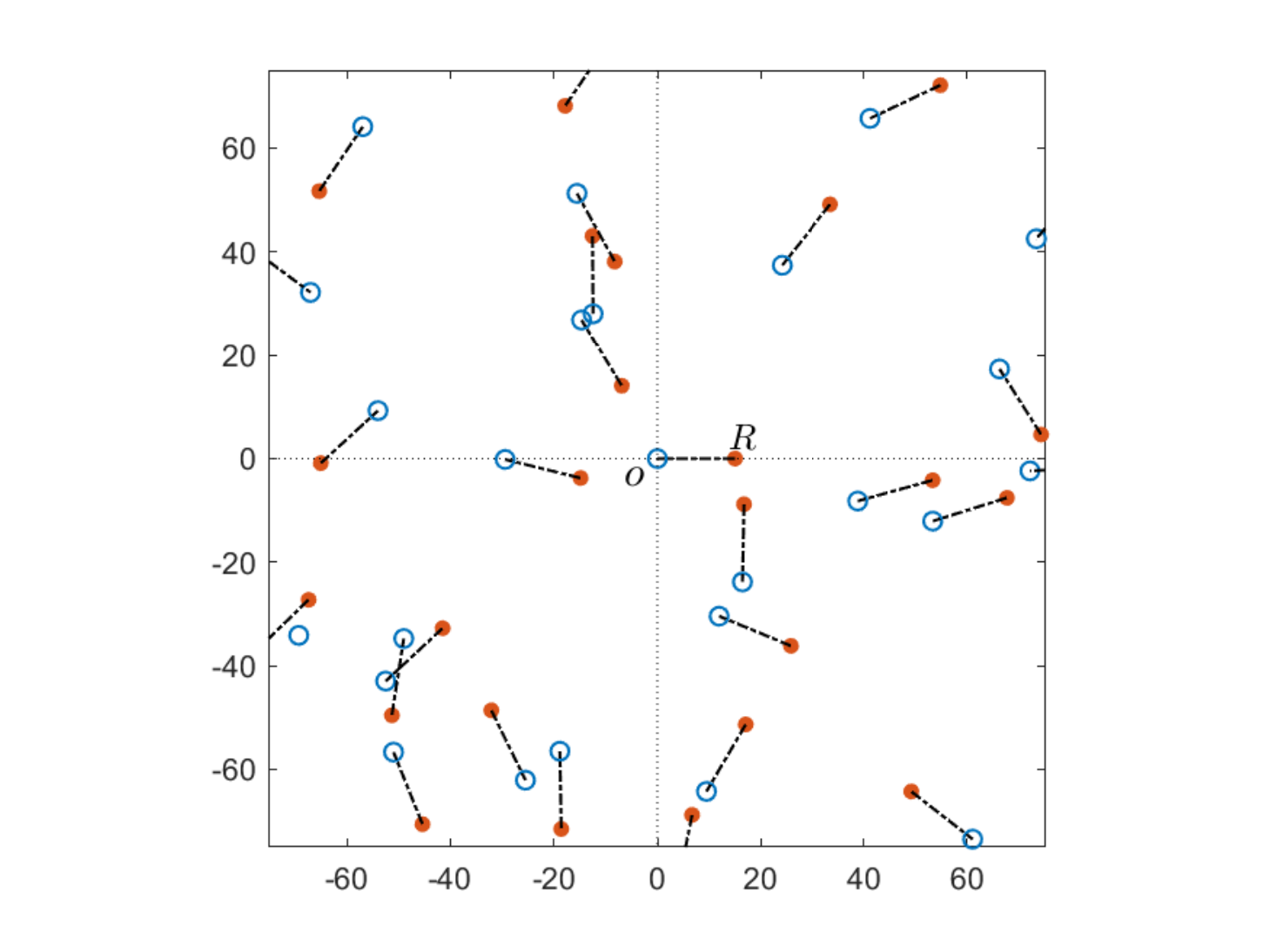}
 \caption{{A typical realization of  the Poisson bipolar network for $\lambda_\text{sd}=10^{-3}$ links/m$^2$ and $R=15$ m}. Orange dots and blue circles represent the locations of sources and destinations, respectively.}
 \label{Fig:Illustration}
\end{figure}
\subsection{Conditional Success Probability}
The transmission is considered to be successful when the received $\sir$ is greater than a threshold $\beta$. 
From \eqref{eq:SIR}, it is clear that the successful transmission probability measured at the typical destination placed at  $o$ depends on the PPP $\Phi$ of the interfering sources  and is given by
{\begin{align}
&\mu_\Phi=\P[\sir_k>\beta|\Phi]
=\mathbb{P}\left[h_{\x_o}^k>\beta R^\alpha \sum\nolimits_{\x\in\Phi}h_{\x}^k\|\x\|^{-\alpha}\nb1(\x\in\Phi_k)\big|\Phi\right],\nonumber\\
&=\mathbb{E}\left[\exp\left(-\beta R^\alpha \sum\nolimits_{\x\in\Phi}h_{\x}^k\|\x\|^{-\alpha}\nb1(\x\in\Phi_k)\right)\big|\Phi\right]=\prod_{\x\in\Phi} \left[\frac{p_{\x}}{1+\beta R^{\alpha}\|\x\|^{-\alpha}} + (1-\p_\x)\right],
\label{eq:conditional_SuccessProb}
\end{align}}
where $p_{\x}$ represents the probability that the source at $\x\in\Phi$ is active. 
{Note that the time average activity $p_\mathbf{x}$ for source at $\x$ is used in \eqref{eq:conditional_SuccessProb}. This implies that the typical destination observes the activities of interfering sources as a time-homogeneous process at any given transmission slot. This assumption will help to develop a new framework for an accurate spatio-temporal analysis, as will be evident shortly. }
{The conditional success probability $\mu_\Phi$ will be useful to determine the packet delivery rate over the typical link for a given $\Phi$. Thus, the knowledge of  the meta distribution, defined below, is  crucial to characterize the queue performance for the typical link.
    \begin{definition}[Meta Distribution] The meta distribution of $\sir$ is defined in \cite{Martin2016Meta} as
    \begin{equation}
        {\rm D}(\beta,x)=\mathbb{P}[\mathbb[\sir_k>\beta|\Phi]>x]=\mathbb{P}[\mu_\Phi>x],
    \end{equation}
    where $\mu_\Phi$, given in (2), is the conditional success probability measured at the typical destination for a given $\Phi$. 
    \end{definition}}
\subsection{Traffic Model and AoI Metric}
\label{eq:Traffic_Model_and_AoI}
{We consider that the  source at $\x\in\Phi$ transmit  updates to  the corresponding  destination regarding its associated physical random processes $H_\x(t)$. }   
{ Each source (independently of others) is assumed to sample its associated  physical random processes  at the beginning of each transmission slot according to a Bernoulli process with parameter $\lambda_\text{a}$.  
This assumption of a fixed parameter  Bernoulli process for modeling the update arrivals  can be seen is an necessary approximation of a  scenario where the  sources may observe different physical random processes. 
 However, one can relax the fixed arrival rate assumption with a few but straightforward modifications to the analysis presented in this paper.}

 { The transmission is considered to be successful if  $\sir$ received at the destination is above threshold $\beta$. 
 Thus, the source is assumed to keep transmitting a status update until it receives the successful transmission acknowledgement from the destination on a separate feedback channel which is assumed to be error free. }
The successful delivery of an update takes a random number of transmissions depending on the channel conditions that further depend on numerous factors, such as fading coefficients, received interference power, and network congestion.
Links that are in close proximity of each other  
may experience arbitrarily small update delivery rate because of severe interference, especially when update arrival rate is high. Therefore, to alleviate the impact of severe interference in such cases, we assume that each source attempts transmission with probability $\xi$ independently of the other sources in a given time slot. 
Also note that  the probability of the attempted  transmission  being  successful in a given time slot  is  the conditional success probability $\mu_\Phi$ because of the assumption of independent fading.
Therefore, the number of slots needed for delivering an update at the typical destination can be modeled using the geometric distribution with parameter $\xi\mu_\Phi$ for a given $\Phi$.

{ 
   Let $t_k$ and $t_k^\prime$ be the instances of the arrival (or, sampling)  and reception of the $k$-th update at the source and destination, respectively.
Given time slot $n$, let $D_n=\max\{k|t_k^\prime \leq n\}$ be the slot index of the most recent  update received at the destination and  $A_n$ be the  slot index  of the arrival (at source) of the most recent update received at destination (i.e., at $D_n$). Thus, for the time-slotted transmission of status updates, the AoI can be defined as  
\begin{equation}
\Delta(n)=\begin{cases}
\Delta(n-1)+1,~~&\text{if~transmission~fails~in~} n\text{-th~slot} \\
\Delta(n-1)+1-A_n,~~&\text{otherwise.} 
\end{cases}\tag{3}
\end{equation}
}
The AoI $\Delta(n)$ increases  {in a staircase fashion} with time and drops upon reception of a new update at the destination to the total  {number of slots} experienced by this new update in the system. Note that the minimum possible AoI is one because we assume arrival and delivery of an update to occur at the beginning and the end of the transmission slots, respectively. 
 Given this background, we are now ready to define PAoI, which will be studied in detail in this paper. 

\begin{definition}[PAoI]
\label{def:Peak_AoI}
The PAoI is defined in \cite{costa2016age} as the value of AoI process $\Delta(n)$ measured immediately before the reception of the $k$-th update and is given by
\begin{equation}
A_k=T_{k-1}+Y_k.
\label{eq:Peak_Age}
\end{equation} 
{where $T_k=t_k^\prime-t_k$ is the time spent by the $k$-th update in the system and $Y_k=t_k^\prime-t_{k-1}^\prime$ is the time elapsed between the receptions of the $(k-1)$-th and $k$-th updates.}
\end{definition}
As evident from the above discussion, the mean PAoI measured at the typical destination depends on the  conditional success probability $\mu_\Phi$ and hence the mean PAoI is a random variable. Therefore, our goal is to determine the distribution of the conditional mean PAoI of the SD pairs distributed across the network. In the following we define the distribution of mean PAoI. 
{\begin{definition}[Conditional mean PAoI]
For a given $\Phi$, the conditional (temporal) mean PAoI measured at the typical destination is defined as
\begin{equation}
\bar{A}(\beta;\Phi)=\E[A_k|\beta;\Phi],
\end{equation}
and the complementary $\cdf$ of $\bar{A}(\beta;\Phi)$ is defined as 
\begin{equation}
\bar{F}(x;\beta) =\mathbb{P}[\bar{A}(\beta;\Phi)>x],
\label{eq:AoI_Distribution}
\end{equation}
where $\beta$ is the $\sir$ threshold.
\end{definition}}
\subsection{Queue Disciplines}
{ As discussed in Section \ref{sec:prior_art}, the construction of the dominant system allows to decouple the activity of typical link with the activities of other links. Thus, for a given $\Phi$, the service rate of the typical link becomes time-invariant (but a lower bound) in the dominant system. This implies that  the service process can be modeled using the geometric distribution for analyzing the conditional  performance bounds of the queue associated with the typical link. Therefore,  with the above discussed Bernoulli sampling, we can model the status update transmissions over the typical link using $\mathtt{Geo/Geo/1}$ with   arrival and service rates equal to $\lambda_{\rm a}$ and $\xi\mu_\Phi$, respectively, for a given $\Phi$.   }
{In particular, we consider $\mathtt{Geo/Geo/1}$ queue with no storage facility (i.e., zero buffer) under preemptive and non-preemptive disciplines.  In preemptive case, the older update in service (or, retransmission) is discarded upon the arrival of a new update. However, in non-preemptive case, the newly arriving updates are discarded until the one in service is successfully delivered.
The  disadvantage of non-preemptive discipline is that if the server takes a long time to transmit the packet (because of failed transmission attempts), the update in the server gets stale, which impacts AoI at the destination. On the contrary,  under the preemptive discipline,  the source always ends up transmitting the most recent update available at the successful transmission. Thus, this discipline  is  optimal from the perspective of minimizing AoI.
Nonetheless, the mean AoIs  under both disciplines are almost the same for the sources experiencing high service rates, as will be evident shortly. That said, the analysis of non-preemptive discipline is still important because it acts as the precursor for the more complicated analysis of preemptive discipline. } 

{Here onward, we will append subscripts ${\rm P}$ and $\rm NP$ to $\bar{A}(\beta;\Phi)$ to denote the conditional mean PAoI under the preemptive and non-preemptive queueing disciplines, respectively. 
 We first present the analysis of  $\bar{A}_{\rm NP}(\beta;\Phi)$ in Section \ref{sec:TypeI}. Next, we extend the analysis $\bar{A}_{\rm P}(\beta;\Phi)$  in Section \ref{sec:TypeII} using the analytical framework developed in Section \ref{sec:TypeI}. }
 
\section{ AoI under non-preemptive $\mathtt{Geo/Geo/1}$ queue}
\label{sec:TypeI}
{In non-preemptive discipline, each source transmits the updates on  the first arrival basis without buffering them.}  As a result, the  updates arriving during the ongoing transmission (i.e., busy server) are dropped. 
The sample path of the  AoI  process for this discipline is illustrated in Fig. \ref{Fig:SamplePath}. 
\begin{figure}[h]
\centering
 \includegraphics[clip, trim=3cm 11.4cm 2cm 7.5cm, width=0.6\textwidth]{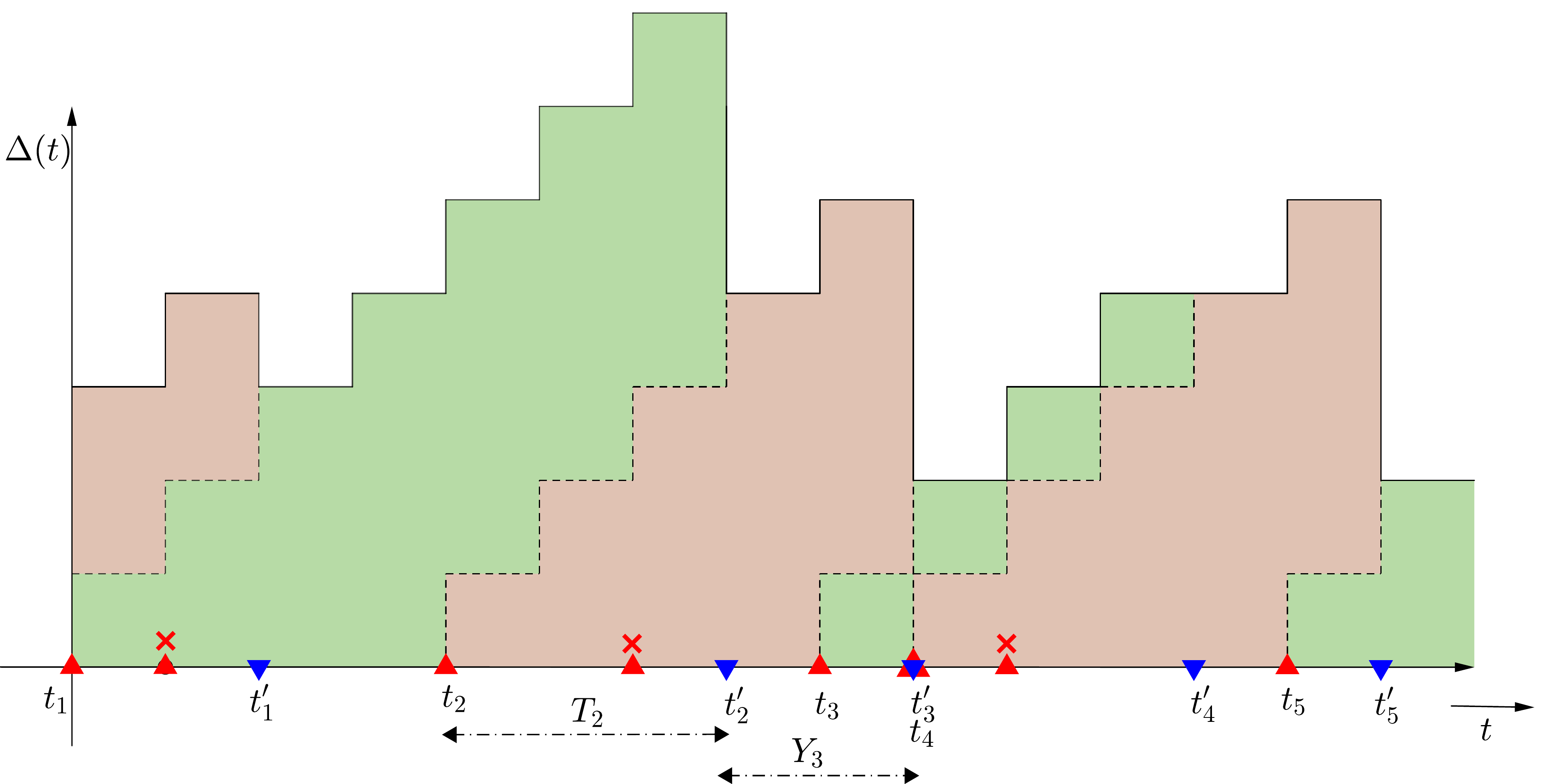}
 \caption{Sample path of the AoI $\Delta(t)$ under the non-preemptive discipline. }
 \label{Fig:SamplePath}
\end{figure}
The red upward and  blue downward arrow marks indicate the arrival  and reception of updates at the source  and destination, respectively. The red cross marks indicate the instances of dropped updates which arrive while the server is busy.
We first derive the conditional mean PAoI   in the following subsection. 
In the subsequent subsections, we will develop an approach to derive the distribution of the conditional mean PAoI using stochastic geometry.
\subsection{Conditional Mean PAoI}
As discussed before, the update delivery rate is governed by the product of the medium access probability $\xi$ and the conditional success probability $\mu_\Phi$. 
Thus, for a given $\Phi$, {the probability mass function ($\pmf$) of number of time slots required for a successful transmission is 
  \begin{equation}
 \P[T_k=n]=\xi\mu_\Phi(1-\xi\mu_\Phi)^{n-1}, \text{~~for~}n=1,2,\dots
 \label{eq:Tk_pmf}
 \end{equation}
Thus, we can obtain $\E[T_k]=\xi^{-1}\mu_\Phi^{-1}$.} 
Since we assumed zero length buffer queue, the next transmission is possible only for  the  update arriving  after the successful reception of the ongoing update. {Therefore, the time between the successful reception of $(k-1)$-th and $k$-th updates is 
\begin{equation*}
Y_k=V_k+T_k,\text{~~~such that~}Y_k\geq 1,
\end{equation*}
where $V_k$ is the number of slots required for the $k$-th update to arrive after the successful delivery of the $(k-1)$-th update. It is worth noting that the inequality $Y_k\geq 1$  follows from the assumption of the transmission of an update begins in the same slot in which it arrives. Also, the inequality $V_k\geq 0$ (necessary to hold $Y_k\geq 1$ since $T_k\geq 1$) is quite evident as the next update  can arrive at the beginning of next transmission slot after successful reception.
Therefore, using the Bernoulli arrival of status updates, we can obtain  ${\rm pmf}$ of $V_k$ as below
     \begin{equation}
         \mathbb{P}[V_k=n]=\lambda_{\rm a}(1-\lambda_{\rm a})^{n}, \text{~~~for~}n=0,1,2,\dots
         \label{eq:Vk_pmf}
     \end{equation}
Using the above ${\rm pmf}$, we can obtain $\mathbb{E}[V_k]={\cal Z}_{\rm a}= \frac{1}{\lambda_{\rm a}}-1$.}
Now, from \eqref{eq:Peak_Age}, the conditional mean PAoI   for given $\Phi$ becomes
\begin{equation}
\bar{A}_{\rm NP}(\beta;\Phi)={\cal Z}_{\rm a}+\frac{2}{\mu_\Phi\xi}.
\label{eq:Mean_Peak_Age_Given_Phi}
\end{equation} 
 Using \eqref{eq:Mean_Peak_Age_Given_Phi} and the distribution of $\mu_\Phi$, we can directly determine the   distribution of $\bar{A}_{\rm NP}(\beta;\Phi)$. However, from \eqref{eq:conditional_SuccessProb}, it can be seen that the knowledge of probability $p_\x$ of the interfering source at $\x\in\Phi$ being active is required to determine the distribution of $\mu_\Phi$. For this, we first determine the activity of the typical source for a given $\mu_\Phi$  in the following subsection which will later be used to define the activities of interfering sources in order to characterize the distribution of   $\mu_\Phi$. 
\subsection{Conditional Activity} 
\label{subsec:ConditionalActivity}
As  each source is assumed to transmit independently in a given time slot with probability $\xi$, the  conditional probability of the  typical source having an  active transmission becomes
\begin{align}
\zeta_o=\xi\mathbf{\pi}_1,\label{eq:Activity}
\end{align}
 where $\mathbf{\pi}_1$ is the conditional steady state probability that the source has  an update to transmit. Thus, $\zeta_o$ depends on the probability $p_\x$ of the interfering source at $\x\in\Phi$ being active through the conditional success probability $\mu_\Phi$ (see \eqref{eq:conditional_SuccessProb}).
 \begin{figure}[h]
\centering
 \includegraphics[clip, trim=17cm 18cm 30cm 14.7cm, width=0.5\textwidth]{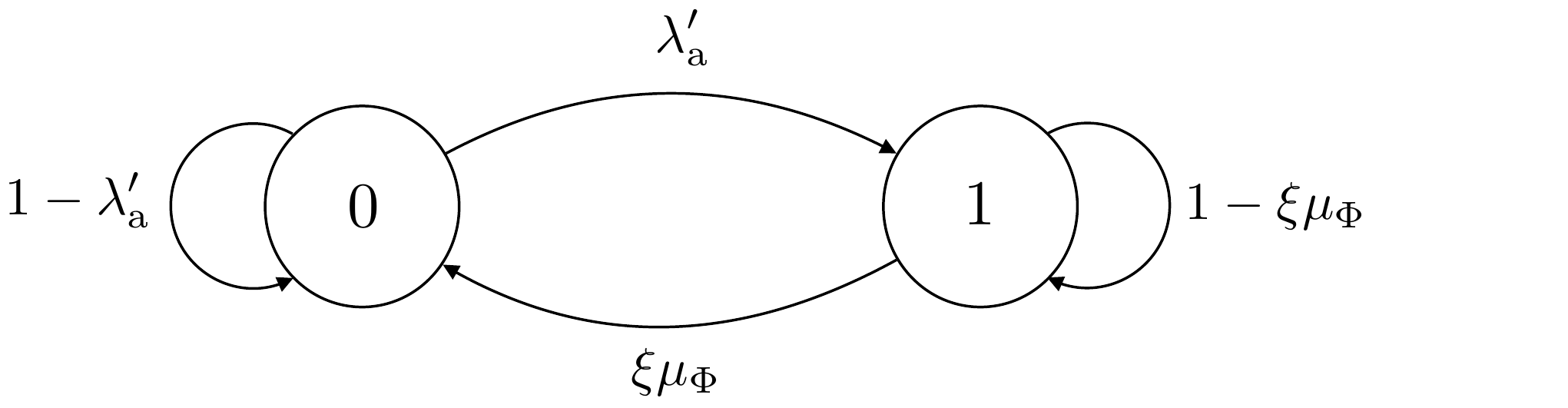}
 \caption{State Diagram of the $\mathtt{Geo/Geo/1}$ queue. States 0 and 1 respectively signify idle and busy states of the typical source.}
 \label{Fig:StateDiagram}
\end{figure}
 The steady state distribution  of a queue is characterized by its arrival and departure processes. In our case, both the arrivals and departures of the updates of $H(t)$ follow geometric distributions with parameters   $\lambda_{\rm a}^\prime=\mathcal{Z}_{\rm a}^{-1}$ and  $\xi\mu_\Phi$, respectively. Fig. \ref{Fig:StateDiagram} shows the state diagram for the $\mathtt{Geo/Geo/1}$ queue. Let $\mathbf{\pi}_0$ and $\mathbf{\pi}_1$ be the steady state probabilities of states 0 and 1, respectively. Thus, we have 
\begin{equation}
\mathbf{\pi}_0=\frac{\xi\mu_\Phi}{\lambda_{\rm a}^\prime +\xi\mu_\Phi} \text{~~and~~} \mathbf{\pi}_1=\frac{\lambda_{\rm a}^\prime}{\lambda_{\rm a}^\prime +\xi\mu_\Phi}.
\label{eq:SteadyState}
\end{equation}
\subsection{Meta Distribution}
\label{subsec:MetaDistribution}
From the above it is clear that the  mean PAoI  jointly depends  on the PPP $\Phi$ of the interfering sources and their activities $p_\x$ through the conditional success probability $\mu_\Phi$. Hence, the knowledge of the exact distribution of $\mu_\Phi$, i.e., $\P[\mu_\Phi\leq x]$, is essential  for characterizing the spatial distribution of the mean PAoI.   However, it is very challenging to capture the temporal correlation among the activities of the sources in the success probability analysis. 
 Therefore, we present the moments and approximate distribution of $\mu_\Phi$ in the following lemma while assuming  the activities $p_\x$ to be independent and identically distributed (i.i.d.). 
\begin{lemma}
\label{lemma:moments_SuccessProb}
 The $b$-th moment of the conditional success probability  $\mu_\Phi$  is 
\begin{align*}
M_b&=\exp\left(-\pi\lambda_\text{sd}\beta^\delta R^2\hat{\delta}C_{\zeta_o}(b) \right),
\numberthis\label{eq:Moments_mu_Phi}
\end{align*}
where $\delta=\frac{2}{\alpha}$, $\hat{\delta}=\Gamma(1+\delta)\Gamma(1-\delta)$ and 
\begin{equation*} 
C_{\zeta_o}(b)=\sum_{m=1}^\infty {b \choose m}{\delta - 1 \choose m-1}\bar{p}_m,
\end{equation*} 
and $\bar{p}_m$ is the $m$-th moment of the activity probability.
The  meta distribution can be approximated with the beta distribution as 
{\begin{align}
{D}(\beta,x)\approx 1-B_{x}(\kappa_1,\kappa_2),
\label{eq:beta_approximation}
\end{align}  
 where $B_x(\cdot,\cdot)$} is the regularized incomplete beta function and
\begin{equation}
\kappa_1=\frac{M_1\kappa_2}{1-M_1}\text{~~and~~}\kappa_2=\frac{(M_1-M_2)(1-M_1)}{M_2-M_1^2}.
\label{eq:Beta_parameters}
\end{equation}
\end{lemma}
\begin{proof}
Please refer to Appendix \ref{app:moments_SuccessProb} for the proof of moments of $\mu_\Phi$. 
\end{proof}
{The calculation of the $\bar{p}_m$ will be presented in the next subsection.} 
It must be noted that the distribution of the conditional success probability $\mu_\Phi$ is approximated using the beta distribution by equating the first two moments, similar to \cite{Martin2016Meta}.  Now, we present an approach for accurate characterization of $\mu_\Phi$ in the following subsection.
\subsection{{A New Approach for Spatiotemporal Analysis}}
\label{subsec:TwoStepAnalysis}
As discussed in Section \ref{subsec:ConditionalActivity}, the activity of the typical  source depends on its successful transmission probability which  further depends on the activities of the other sources in $\Phi$ through interference (refer to \eqref{eq:conditional_SuccessProb}).  Besides, the transmission of the typical  source causes strong interference to its neighbouring sources which in turn affects their activities. Hence, the successful transmission probabilities  of the typical source and its neighbouring sources are correlated  through interference which introduces coupling between the operations of their associated queues.
 { As discussed in Section \ref{sec:prior_art}, the spatiotemproal analysis under the correlated queues is complex. Hence, we adopt to the approach of constructing a dominant system for analyzing the performance bound on conditional success probability and thereby on the conditional mean PAoI.} 

On the similar lines of \cite{bonald2004wireless}, we present a novel two-steps analytical framework to enable an accurate success probability analysis using stochastic geometry while accounting for the temporal correlation in the queues associated with the SD pairs.\\
\begin{figure*}
\centering
 \includegraphics[trim=0cm 0cm 0cm 0cm, width=.5\textwidth]{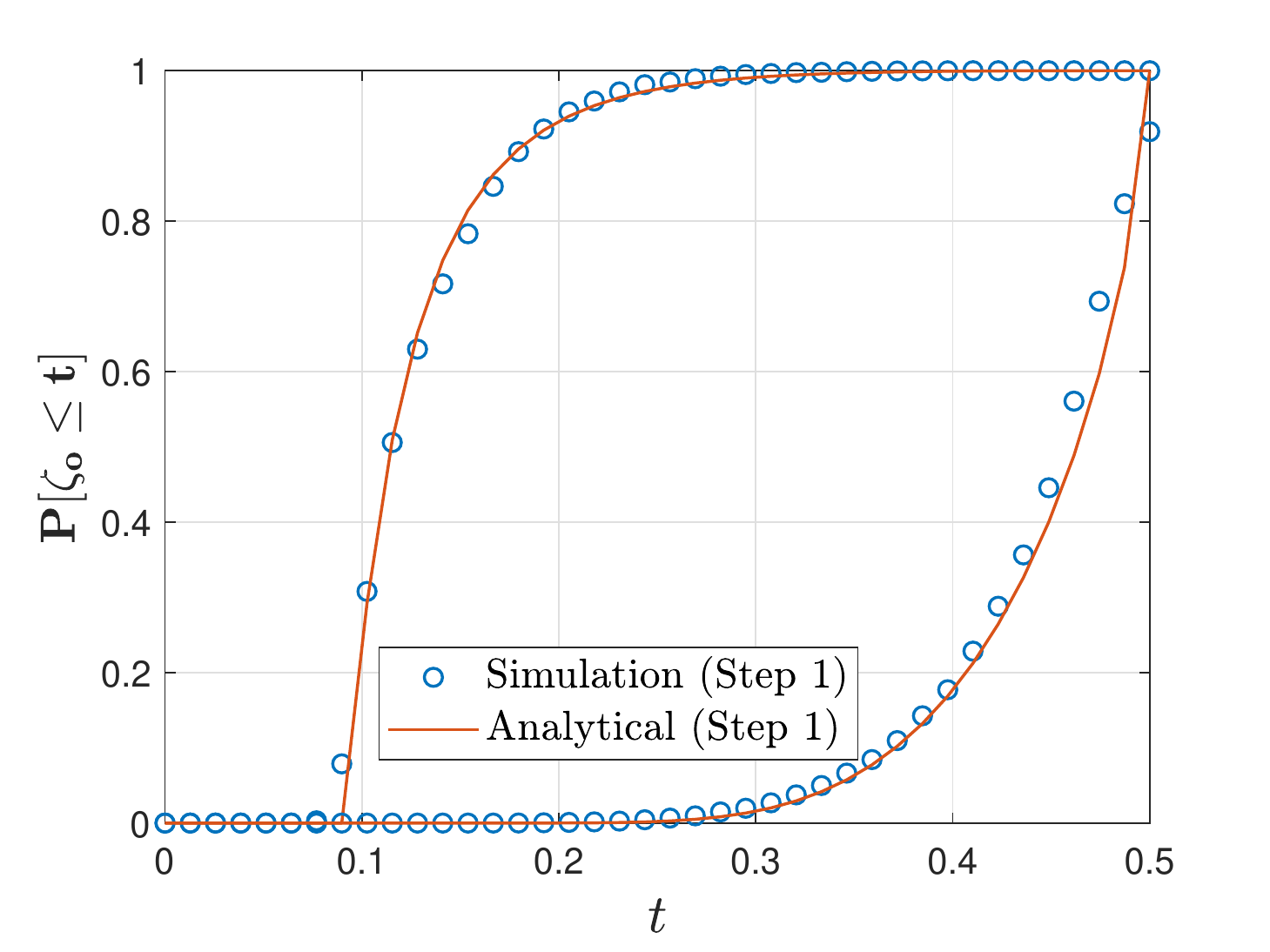}
 \caption{CDF of the activity in the dominant system for {$\lambda_{\rm a}=0.1$, $\xi=0.5$,  $\beta=2$,  $\alpha=4$, $R=10$, and $\lambda_{\rm sd}=\{10^{-3},10^{-4}\}$}.}
  \label{Fig:CDF_Activity_DomSys}
\end{figure*}
$\bullet$ {\em Step 1 ({\em Dominant System}):}
  For the dominant system,  the interfering sources having no updates to transmit are assumed to transmit dummy packets  with probability $\xi$. As a result, the success probability measured at the typical destination is  a lower bound to that  in the original system. The $b$-th moment and approximate distribution of $\mu_\Phi$ for the dominant system can be evaluated using Lemma \ref{lemma:moments_SuccessProb}  by setting $\bar{p}_m=\xi^m$.
Using \eqref{eq:Activity},  we can obtain the distribution of the activity of the typical source in the dominant system  as
\begin{align*}
\P[\zeta_o\leq t]=\P\left[\lambda_{\rm a}^\prime\left(\frac{1}{t}-\frac{1}{\xi}\right)\leq \mu_\Phi \right]\approx1-{B_{\lambda_{\rm a}^\prime\left(\frac{1}{t}-\frac{1}{\xi}\right)}}\left(\kappa_1,\kappa_2\right),\numberthis\label{eq:activity_CDF_Dom_Sys}
\end{align*}  
for $0<t\leq \xi$, where $\kappa_1$ and $\kappa_2$ are evaluated using \eqref{eq:Beta_parameters} by setting  $\bar{p}_m=\xi^m$. It is quite evident that the activity $\zeta_o$ is less than  $\xi$ which is also consistent with the  assumption  of setting $p_\x=\xi$ for $\forall \x\in\Phi$ to define the dominant system. Fig. \ref{Fig:CDF_Activity_DomSys} illustrates the accuracy of the above approximate distribution of the activity of the typical source under the dominant system.
Using \eqref{eq:activity_CDF_Dom_Sys}, the $m$-th moment of the activity of the typical source in this dominant system can be evaluated as
\begin{align}
\bar{p}_m^{\rm D}=m\int_0^\xi t^{m-1}\P[\zeta_o>t]{\rm d}t=m\int_0^\xi t^{m-1} {B_{\lambda_{\rm a}^\prime\left(\frac{1}{t}-\frac{1}{\xi}\right)}
}(\kappa_1,\kappa_2) {\rm d}t,
\label{eq:Activity_Moments_DomSys}
\end{align}
where the $\kappa_1$ and $\kappa_2$ are evaluated using \eqref{eq:Beta_parameters} by setting  $\bar{p}_m=\xi^m$.\\
$\bullet$ {\em Step 2 (Second-Degree of System Modifications):} Inspired by \cite{bonald2004wireless}, we construct a second-degree system in which each interfering source is assumed to operate in the dominant system described in Step 1 (i.e., the interference field seen by a given interfering source is constructed based on Step 1). The typical SD link is now assumed to operate under these interfering courses. 
Naturally, the activities of the interfering sources  will be higher in this modified system compared to those in the original system. As a result, the typical SD link will experience increased interference, and hence its conditional success probability will be a lower bound to that in the original system. It is easy to see that activities of the sources (in their dominant systems) are identically but not independently distributed. However, as is standard in similar stochastic geometry-based investigations, we will assume them to be independent to make the analysis tractable. In other words, we model the activities of the interfering sources in this modified system independently using the activity distribution presented in \eqref{eq:activity_CDF_Dom_Sys} for the typical SD link in its dominant system. As demonstrated in Section V, this assumption does not impact the accuracy of our results.
Hence, similar to Step 1, the $b$-th moment and the approximate distribution of $\mu_\Phi$ for this second-degree modified system can be determined  using Lemma \ref{lemma:moments_SuccessProb}  by setting $\bar{p}_m=\bar{p}_m^{\rm D}$. 

%
\subsection{Moments and Distribution of $\bar{A}_{\rm NP}(\beta;\Phi)$ }
Here, we derive the bounds on the moments and distribution of $\bar{A}_{\rm NP}(\beta;\Phi)$ using the two-step analysis of conditional success probability presented in Sections \ref{subsec:MetaDistribution} and \ref{subsec:TwoStepAnalysis}.

\begin{theorem}
The upper bound of the $b$-th moment of the conditional mean PAoI for  {non-preemptive queuing discipline with no storage} is
\begin{align}
P_b=\sum_{n=0}^b  {b \choose n} \ncalZ_{\text{a}}^{b-n}2^{n}\xi^{-n}\exp\left(-\pi\lambda_\text{sd}\beta^\delta R^2\hat{\delta}C_{\zeta_o}(-n) \right),\numberthis\label{eq:Moment_Peak_AoI}
\end{align}
where 
\begin{align*} 
C_{\zeta_o}(-n)=\sum_{m=1}^\infty {-n \choose m}{\delta - 1 \choose m-1}\bar{p}_m^{\rm D},
\end{align*} 
and $\bar{p}_m^{\rm D}$ is given in \eqref{eq:Activity_Moments_DomSys}. 
 \end{theorem}
\begin{proof}
Using \eqref{eq:Mean_Peak_Age_Given_Phi},  the $b$-th moment of the conditional mean  PAoI can be determined as
 \begin{align*}
P_b&=\E_\Phi\left[(2(\xi\mu_\Phi)^{-1}+{\cal Z}_{\rm a})^b\right]\stackrel{(a)}{=}\sum_{n=0}^b  {b \choose n} \ncalZ_{\text{a}}^{b-n}2^{n}\xi^{-n}M_{-n} \numberthis\label{eq:Moment_Peak_AoI_a}
 \end{align*}  
where (a) follows using the binomial expansion and $\E[\mu_\Phi^{-n}]=M_{-n}$. According to the Step 2 discussed in Subsection \ref{subsec:TwoStepAnalysis},  $M_{-n}$ can be obtained using Lemma \ref{lemma:moments_SuccessProb} by setting  $\bar{p}_m=\bar{p}_m^{\rm D}$. Recall that the two-step analysis provides a lower bound on the success probability $\mu_\Phi$ because of overestimating the activities of the interfering sources. Therefore, the $b$-th moment of $\bar{A}_1(\beta;\Phi)$ given in \eqref{eq:Moment_Peak_AoI} is indeed  an upper bound since $\bar{A}_1(\beta;\Phi)$ is inversely proportional to $\mu_\Phi$.      
\end{proof}
In the following corollary, we present the simplified expressions for the evaluation of the first two moments of $\bar{A}_1(\beta;\Phi)$.
\begin{cor}
\label{cor:AoI_M12_Q1}
The upper bound of the first two moments of the conditional mean PAoI for { non-preemptive queuing discipline with no storage} are
\begin{align*}
P_1&=\ncalZ_{\text{a}} + 2\xi^{-1}M_{-1}\numberthis\label{eq:Mean_Peak_AoI}\\
P_2&=\ncalZ_{\text{a}}^2 + 4 \ncalZ_{\text{a}}\xi^{-1}M_{-1} + 4\xi^{-2}M_{-2} \numberthis \label{eq:M2_Peak_AoI}
\end{align*}
and the upper bound of its variance is 
\begin{align*}
\mathtt{Var}&=  4\xi^{-2}\left(M_{-2}-M_{-1}^2\right),\numberthis \label{eq:Variance_Peak_AoI}
\end{align*} 
where $M_{l}=\exp\left(-\pi\lambda_\text{sd}\beta^\delta R^2\hat{\delta}C_{\zeta_o}(l) \right)$, distribution of $\zeta_o$ is given in \eqref{eq:activity_CDF_Dom_Sys}, and
\begin{align*}
C_{\zeta_o}(-1)&=-\E\left[\zeta_o(1-\zeta_o)^{\delta-1}\right]\\
\text{and~}
C_{\zeta_o}(-2)&=(\delta-1)\E\left[\zeta_o(1-\zeta_o)^{\delta-2}\right] - (\delta+1)\E\left[\zeta_o(1-\zeta_o)^{\delta-1}\right].
\end{align*}
\end{cor}
\begin{proof}
Please refer to Appendix \ref{app:AoI_M12_Q1} for the proof.
\end{proof}
  \begin{remark}
For the mean PAoI given in \eqref{eq:Mean_Peak_AoI}, the first term captures the impact of the update arrival rate $\lambda_\text{a}$ whereas the second term depends on the inverse mean of the conditional success probability, which captures the impact of the wireless link parameters such as the link distance $R$, network density  $\lambda_\text{sd}$, and path-loss exponent $\alpha$. 
 { Furthermore, it is worth noting that the variance of the temporal mean AoI is independent of the arrival rate and entirely depends on the link quality parameters.  This is because the arrival rate is assumed to be the same for all SD links and it corresponds to the additive term in conditional mean PAoI given in \eqref{eq:Mean_Peak_Age_Given_Phi}. }
\end{remark}
Now,  using the beta approximation of the conditional success probability presented in Lemma \ref{lemma:moments_SuccessProb}, we determine the distribution of $\bar{A}_{\rm NP}(\beta;\Phi)$  in the following corollary.
\begin{cor}
\label{cor:Distribution_AoI_I}
{Under the beta approximation, the  $\cdf$ of the conditional mean PAoI for {non-preemptive queueing discipline} is
\begin{align*}
{F}(x;\beta)= 1-{B_{2\xi^{-1}(x-\mathcal{Z}_{\rm a})^{-1}}}(\kappa_1,\kappa_2),
\numberthis\label{eq:AoIDistribition_1}
\end{align*}   }
where  $\kappa_1$ and $\kappa_2$ are obtain using Lemma \ref{lemma:moments_SuccessProb} by setting $\bar{p}_m=\bar{p}_m^{\rm D}$.
\end{cor}
\begin{proof}
Using \eqref{eq:Mean_Peak_Age_Given_Phi} and the beta approximation of the distribution of $\mu_\Phi$ given in Lemma \ref{lemma:moments_SuccessProb}, we can determine the distribution of $\bar{A}_1(\beta;\Phi)$ as
\begin{align*}
\P[\bar{A}_{\rm NP}(\beta;\Phi)\leq x]=\P\left[\mu_\Phi\geq 2\xi^{-1}(x-\mathcal{Z}_{\rm a})^{-1}\right]{={\rm D}(2\xi^{-1}(x-\mathcal{Z}_{\rm a})^{-1})}.
\end{align*} 
 Further, using  the beta approximation for the distribution of $\mu_\Phi$ (given in Lemma \ref{lemma:moments_SuccessProb}), we obtain \eqref{eq:AoIDistribition_1}. The parameters $\kappa_1$ and $\kappa_2$ need to be determined for  $\bar{p}_m=\bar{p}_m^{\rm D}$ (refer to  Section \ref{subsec:TwoStepAnalysis}). 
\end{proof}

\section{ AoI {under} preemptive $\mathtt{Geo/Geo/1}$ queue}
\label{sec:TypeII}
In preemptive discipline, each source  transmits the most recent update available at the source in a given transmission slot.  As a result, this  queueing discipline  helps to reduce the AoI which is clearly significant when the update arrival rate is high while the delivery rate is low. 
This discipline is, in fact, optimal from the perspective of minimizing AoI at the destination as the source always ends up transmitting its most recent update arrival.  
 A representative sample path of the  AoI process  under the preemptive queue discipline is presented in Fig. \ref{Fig:SamplePath_2}.  The red upward  and blue downward arrow marks indicate the time instants 
 of update arrivals at the source and deliveries at  the destination, respectively.   The red plus sign marks indicate the instances of replacing the older update with a newly arrived update and $t_{k,n}$ highlighted in red indicates the $n$-th replacement of the $k$-th update. 
\begin{figure}[h]
\centering
 \includegraphics[clip, trim=6.1cm 8.8cm 1cm 6.5cm, width=0.6\textwidth]{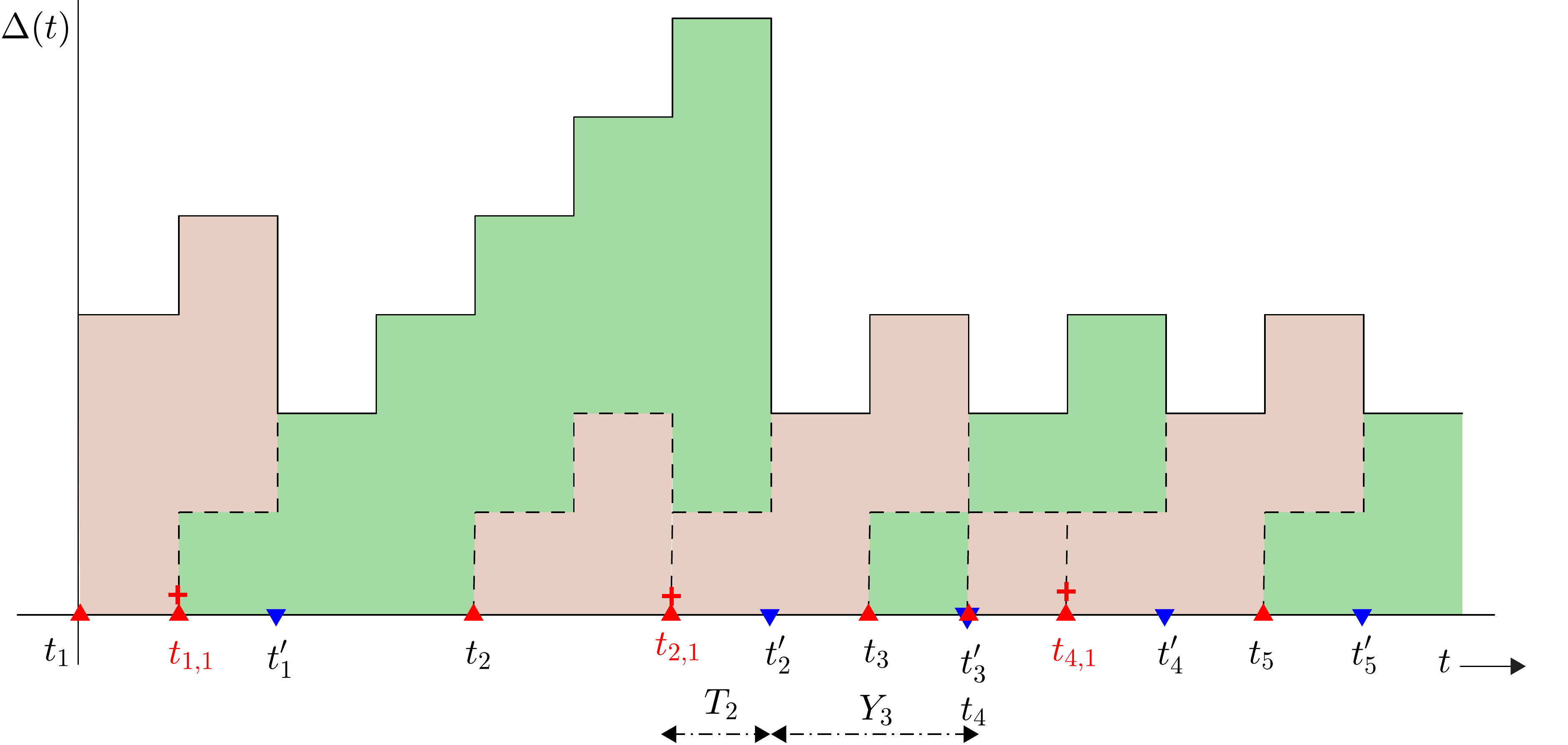}
 \caption{Sample path of the AoI $\Delta(t)$ under the preemptive queue discipline.}
 \label{Fig:SamplePath_2}
\end{figure}
\subsection{Conditional Mean PAoI }
In this subsection, we derive the conditional PAoI for the typical destination for a given $\Phi$. For this, the primary step is to obtain the means of the inter-departure time $Y_k$ and service time $T_k$ for a given conditional success probability $\mu_\Phi$ (measured at the typical destination). While it is quite straightforward to determine the mean of $Y_k$, the derivation of the mean of $T_k$ needs careful consideration of  the successive replacement of updates until the next successful transmission occurs.

{\em Mean of $Y_k$:} Recall that the expected time for a new arrival is {${\cal Z}_{\rm a}=\frac{1}{\lambda_{\rm a}}-1$}. Also, recall that the update delivery rate is $\xi\mu_\Phi$ which follows from the fact that both the successful transmission and the medium access  are independent events across transmission slots. Thus, since  the transmission starts from the new arrival of the update after the successful departure, the mean time between two departures simply becomes
\begin{equation}
\E[Y_k]={\cal Z}_{\rm a}+\frac{1}{\xi\mu_\Phi}.
\label{eq:Mean_Yk_2}
\end{equation}

{\em Mean and Distribution of Delivery Times:} The  delivery times  restart after every new update arrival before epochs of successful delivery. \begin{figure}[h]
\centering
\includegraphics[clip, trim=8.3cm 10.3cm 10cm 17.2cm, width=0.65\textwidth]{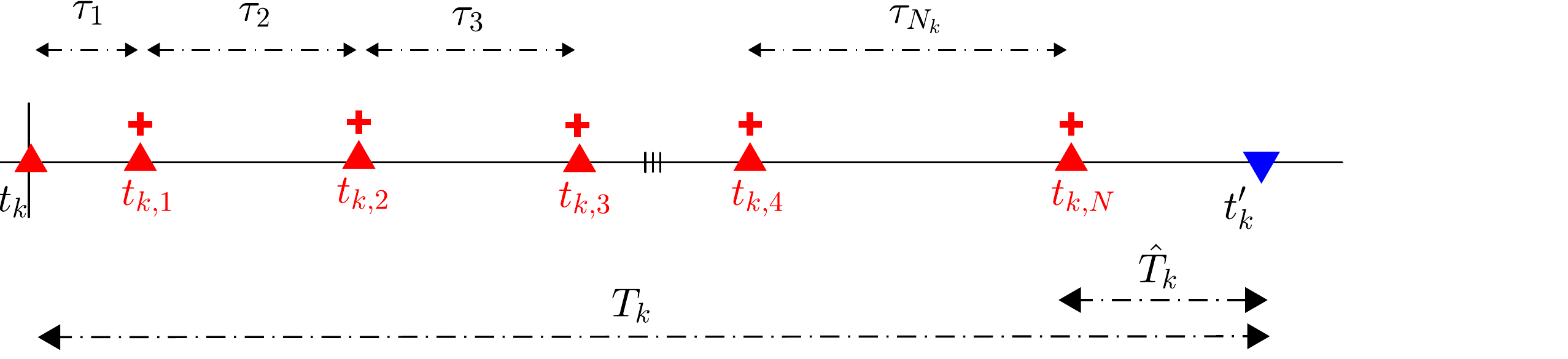}
 \caption{Illustration of the process $\hat{T}_k$ under preemption.}
 \label{Fig:Illustration_Tk}
\end{figure} 
 { Let $\hat{T}_k$ be the delivery times of the latest update under preemptive queueing disciple.} Fig. \ref{Fig:Illustration_Tk} illustrates the typical realization of random process $\hat{T}_k$.
The number of slots required to deliver the latest update, i.e.  $\hat{T}_k$, is given by
\begin{equation*}
\hat{T}_k=T_k-X_{N_k},
\end{equation*}
 where $X_{N_k}=\sum_{i=1}^{N_k}\tau_i$, and $N_k$ is the number of new update arrivals occurring between the arrival of the $k$-th update (right after ($k-1$)-th successful transmission) and the $k$-th successful transmission, and  $\tau_n$ is the number of slots between the arrivals of two successive updates. 
{ \begin{lemma}
 \label{lemma:Tk_pmf}
 For a zero buffer queue with preemption, the $\pmf$ of number of slots required to deliver the latest update follows a geometric distribution as given below
 \begin{equation}
 \P[\hat{T}_k=m]=
 q_{\rm s}(1-q_{\rm s})^{m-1},~~~~~~\text{~for~}m=1,2,\dots
\label{eq:pmf_Tk}
 \end{equation}
 where $q_{\rm s}=\xi\mu_\Phi + \lambda_\text{a}(1-\xi\mu_\Phi)$.
 \end{lemma}}
 \begin{proof}
 Please refer to the Appendix \ref{app:Tk_pmf} for the proof. 
 \end{proof}
{
Now, the following theorem presents the moment generating function ($\mgf$) of the conditional PAoI $A_k$ for a given $\Phi$.
\begin{theorem}
\label{thm:Peak_AoI_queue2_mgf}
For a zero buffer queue with preemption, the $\mgf$ of the conditional PAoI   is 
\begin{equation}
    \ncalM_{{\rm P},A_k}(t)=\frac{\lambda_{\rm a}\xi\mu_\Phi q_{\rm s}e^{2t}}{(1-(1-\lambda_{\rm a})e^t)(1-(1-\xi\mu_\Phi)e^t)(1-(1-q_{\rm s})e^t)},
    \label{eq:Peak_conditional_AoI_mgf}
\end{equation}
and the mean of the conditional PAoI is
\begin{align*}
\bar{A}_{\rm P}(\beta;\Phi)&=\ncalZ_{\text{a}}+\frac{1}{\xi\mu_\Phi} + \frac{1}{q_{\rm s}},\numberthis\label{eq:Conditional_Mean_Peak_AoI_2}
\end{align*} 
 where $q_{\rm s}=\xi\mu_\Phi + \lambda_\text{a}(1-\xi\mu_\Phi)$.
\end{theorem}
\begin{proof}
From Definition \ref{def:Peak_AoI}, the conditional PAoI can be written as $$A_k=\hat{T}_{k-1}+Y_k=\hat{T}_{k-1}+V_k+T_k.$$ Since $\hat{T}_{k-1}$, $V_k$, and $T_k$ are independent of each other and also themselves across $k$, the $\mgf$ of $A_k$ for a given $\Phi$ becomes
\begin{equation*}
    \ncalM_{A_k}(t)=\ncalM_{\hat{T}_k}(t)\ncalM_{A_k}(t)\ncalM_{T_k}(t).
\end{equation*}
We have $\hat{T}_k\sim\mathtt{Geo}(q_{\rm s})$ for $\hat{T}_k\geq 1$ (from Lemma \ref{lemma:Tk_pmf} ),  $V_k\sim\mathtt{Geo}(\lambda_{a})$ for $V_k\geq 0$ (from \eqref{eq:Vk_pmf}), and $T_k\sim\mathtt{Geo}(\xi\mu_\Phi)$ for $T_k\geq 1$ (from \eqref{eq:Tk_pmf}). Therefore,  \eqref{eq:Peak_conditional_AoI_mgf} follows by substituting the $\mgf$s of geometric distributions with above parameters. 
Next, \eqref{eq:Conditional_Mean_Peak_AoI_2} directly follows by substituting the means of $\hat{T}_{k}$, $V_k$, and $T_k$ in $A_k$.
\end{proof}}
\begin{remark} 
{ From \eqref{eq:Mean_Peak_Age_Given_Phi} and \eqref{eq:Conditional_Mean_Peak_AoI_2}, it is evident that the conditional mean PAoI observed by the typical SD link with a preemptive queue is strictly less than that with a  non-preemptive queue  for any $\lambda_{\rm a}$ as long as $\xi\mu_\Phi<1$. } 
From \eqref{eq:Mean_Peak_Age_Given_Phi} and \eqref{eq:Conditional_Mean_Peak_AoI_2}, we can verify that 
\begin{align*}
\lim_{\lambda_{\rm a}\to 1}\bar{A}_{\rm NP}(\beta;\Phi)=\frac{2}{\xi\mu_\Phi}\text{~~and~~}\lim_{\lambda_{\rm a}\to 1}\bar{A}_{\rm P}(\beta;\Phi)=1+\frac{1}{\xi\mu_\Phi}.
\end{align*}
From this, we can say that   preemptive  discipline reduces the mean  peak  AoI almost by a factor of two compared to the non-preemptive discipline in the asymptotic regime of $\lambda_{\rm a}$ when $\xi\mu_\Phi$ is low.
{  Further, we can also verify that 
    \begin{equation*}
\lim_{\xi\mu_\Phi\to 1}\bar{A}_{\rm NP}(\beta;\Phi)=\lim_{\xi\mu_\Phi\to 1}\bar{A}_{\rm P}(\beta;\Phi)={\cal Z}_{\rm a}+2.
\end{equation*}
    This implies that the sources observing high successful transmission probability can select any one of these queueing disciplines without much compromising on the mean PAoI when $\xi=1$. Besides, note that this  performance analysis also holds for the case where the sources can  selectively opt for any one of these queueing disciplines. This is essentially because of the same transmission activities for any given source under both of these queues, as will discussed next. }
\end{remark}

\subsection{Activity and Distribution of $\mu_\Phi$}
\label{subsec:Activity_Dist_mu}
As the preemptive queue just replaces the older updates with the newly arriving updates during the retransmission instances, its transmission rate  is the same as that of  the  non-preemptive queue. In fact, the state diagrams for both these queues are equivalent (please refer to Fig. \ref{Fig:StateDiagram}). As a result,  the conditional steady state distributions (and thus the conditional activities),  for a given  $\mu_\Phi$, are also the same for both of these queues and are given by \eqref{eq:SteadyState}.  

 The  point process of  interferers is  characterized by their activities, which is the same  for both queues. As a result, the distributions of the conditional success probability $\mu_\Phi$ must also be the same for both of these queues. Based on these arguments, the moments and approximate  distribution of the conditional success probability $\mu_\Phi$ presented in Lemma \ref{lemma:moments_SuccessProb} along with the distribution of transmission activities of interfering sources given in \eqref{eq:activity_CDF_Dom_Sys} can be directly extended for the  analysis of AoI under preemptive queue. 
\subsection{Moments and Distribution of $\bar{A}_{\rm P}(\beta;\Phi)$ }  
In this subsection, we  derive  bounds on the $b$-th moment and the distribution of $\bar{A}_{\rm P}(\beta;\Phi)$ using the two-step analysis of conditional success probability presented in Section \ref{subsec:TwoStepAnalysis}.  
\begin{theorem}
\label{thm:AoI_Mb_Q2}
The upper bound of the $b$-th moment of the conditional mean PAoI for  { preemptive queueing discipline with no storage }  is
\begin{align}
P_b=\sum_{l+m+n=b}{b\choose l,m,n} {\cal Z}_{\rm a}^l\xi^{-m}S(n;m),
\label{eq:TypeII_Qb}
\end{align}
where { ${b\choose l,m,n}=\frac{b!}{l!m!n!}$},
\begin{align}
S(n;m)&=\sum_{k=0}^\infty {n+k-1 \choose k} (1-\lambda_{\rm a})^k\sum_{l=0}^k (-1)^l{k \choose l} \xi^lM_{l-m},\label{eq:Snm}
\end{align}
and $M_l$ is evaluated using Lemma \ref{lemma:moments_SuccessProb} by setting $\bar{p}_m=\bar{p}_m^{\rm D}$ which is given in \eqref{eq:Activity_Moments_DomSys}.  
 \end{theorem}
\begin{proof}
Please refer to Appendix \ref{app:AoI_Mb_Q2} for the proof.
 \end{proof}
 
The general result presented above can be used to derive simple expressions for the  first two moments of $\bar{A}_2(\beta;\Phi)$, which are presented next.
 \begin{cor}
\label{cor:AoI_M12_Q2}
The upper bound of the first two moments of the conditional mean PAoI for { preemptive queueing discipline with no storage } are
 \begin{align}
P_1&={\cal Z}_\text{a}+\xi_{-1}M_{-1}+S(1;0),\label{eq:Mean_AoI_2}\\
P_2&=\ncalZ_{\text{a}}^2 + 2\ncalZ_{\text{a}}\xi^{-1}M_{-1}+\xi^{-2}M_{-2} + 2{\cal Z}_{\rm a}  S(1;0)+ 2\xi^{-1}S(1;1)+S(2;0),\label{eq:M2_AoI_2}
\end{align} 
where {$S(n;m)$ is given in \eqref{eq:Snm}} 
and $M_l=\exp\left(-\pi\lambda_\text{sd}\beta^\delta R^2\hat{\delta}C_{\zeta_o}(l)\right)$, such that
\begin{align*}
C_{\zeta_o}(l)=\begin{cases}
-\E\left[\zeta_o(1-\zeta_o)^{\delta-1}\right], &\text{for~}l=-1,\\
(\delta-1)\E\left[\zeta_o(1-\zeta_o)^{\delta-2}\right] - (\delta+1)\E\left[\zeta_o(1-\zeta_o)^{\delta-1}\right],&\text{for~}l=-2,\\
\sum_{m=1}^l {l \choose m}{\delta - 1 \choose m-1} \bar{p}_{m}^{\rm D}, &\text{for~}l>0,
\end{cases}
\end{align*} 
and the distribution of $\zeta_o$ is given in \eqref{eq:activity_CDF_Dom_Sys} and $\bar{p}_m^{\rm D}$ is given in \eqref{eq:Activity_Moments_DomSys}.   
 \end{cor}
 \begin{proof}
Please refer to Appendix \ref{app:AoI_M12_Q2} for the proof.
 \end{proof}
The following corollary presents  an approximation of distribution of $\bar{A}_{\rm P}(\beta;\Phi)$.
 \begin{cor}
 \label{cor:Distribution_AoI_II}
Under the beta approximation, the $\cdf$ of  the conditional mean PAoI for { preemptive queueing discipline} queue is
\begin{align*}
{F}(x;\beta)&= {B_{g_{\mu_\Phi}(x)}}(\kappa_1,\kappa_2)
\numberthis\label{eq:AoIDistribition_2}
\end{align*} 
where $g_{\mu_\Phi}(x)=\left\{\mu_\Phi\in[0,1]:\bar{A}_2(\beta;\Phi)=x\right\},$
and $\kappa_1$ and $\kappa_2$ are obtain using Lemma \ref{lemma:moments_SuccessProb} by setting $\bar{p}_m=\bar{p}_m^{\rm D}$.
\end{cor}
\begin{proof}
Let $g_{\mu_\Phi}=\bar{A}_2(\beta;\Phi)$ is a function of $\mu_\Phi$. Therefore,   the $\cdf$ of $\bar{A}_2(\beta;\Phi)$ becomes
\begin{align*}
\P[\bar{A}_{\rm P}(\beta;\Phi)\leq x]=\P[\mu_\Phi\leq g_{\mu_\Phi}(x)]{=1-{\rm D}(g_{\mu_\Phi}(x))},
\end{align*}  
where $g_{\mu_\Phi}(x)=\left\{\mu_\Phi\in[0,1]:\bar{A}_2(\beta;\Phi)=x\right\}$.
Next, using the beta approximation for the distribution of $\mu_\Phi$ given in Lemma \ref{lemma:moments_SuccessProb}, we obtain \eqref{eq:AoIDistribition_2}. 
 The parameters $\kappa_1$ and $\kappa_2$ of the beta approximation are obtained using  the two-step analysis by setting  $\bar{p}_m=\bar{p}_m^{\rm D}$ in \eqref{eq:Beta_parameters}.
\end{proof}
\section{Numerical Results and Discussion}
\label{sec:Results}
This paper presents a new approach of the two-step system level modification for enabling the  success probability analysis when the queues at the SD pairs are correlated. Therefore, before presenting the numerical analysis of the PAoI, we verify the application of  this new two-step analytical framework for characterizing the AoI using simulation results in the following subsection. Throughout this section, we consider the system parameter as  $\lambda_a=0.3$ updates/slot, $\lambda_\text{sd}=10^{-3}$ links/m$^2$, $\xi=0.5$, $R=15$ m,  $\beta=3$ dB, and $\alpha=4$ unless mentioned otherwise. { In the figures, $Q_1$ and $Q_2$ indicate the curves correspond to the  non-preemptive and preemptive queiueing disciplines, respectively.}
\begin{figure*}
\centering
\hspace{-5mm}\includegraphics[trim=0cm 0cm 0cm 0cm, width=.36\textwidth]{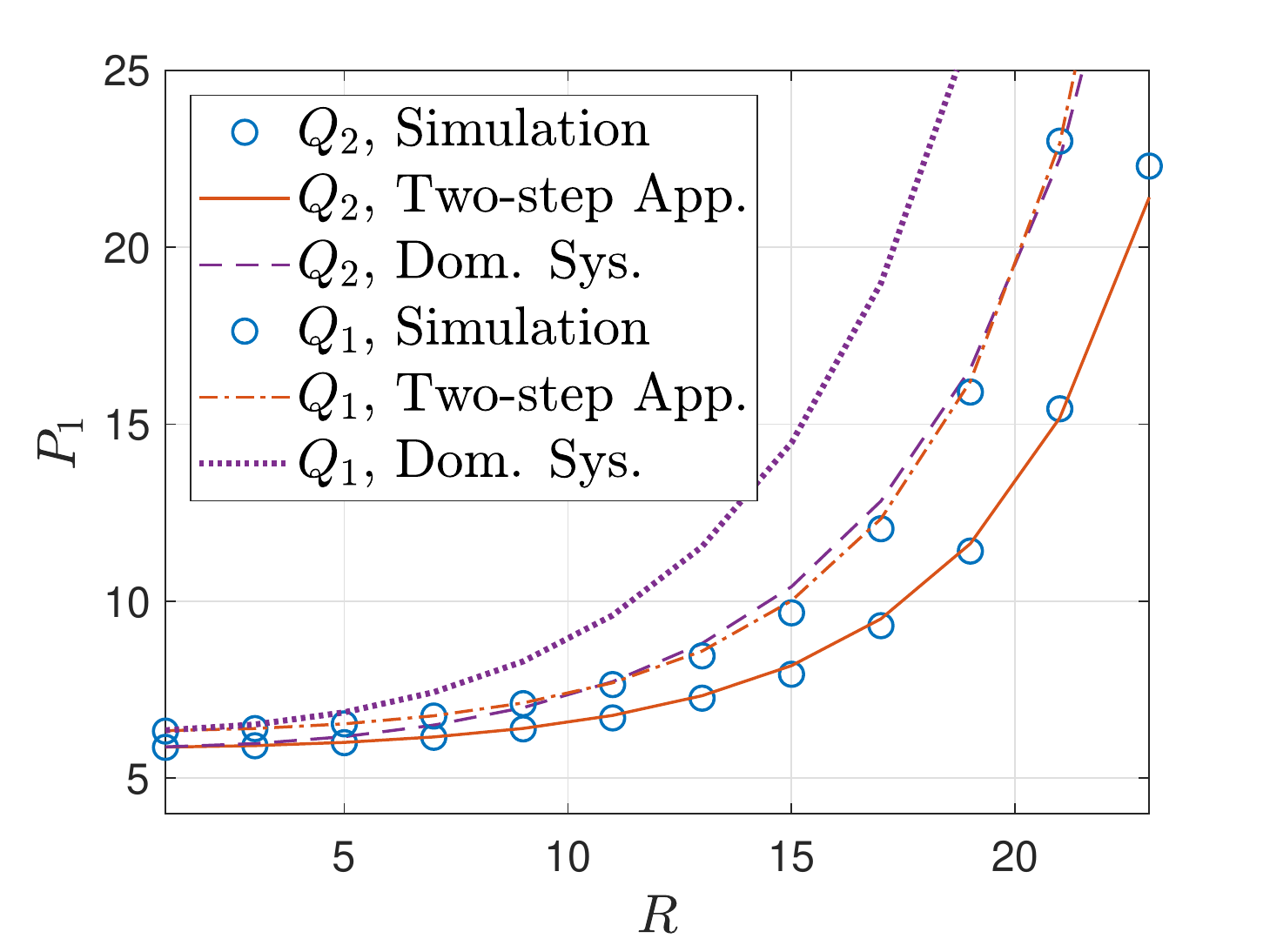}
\hspace{-5mm}\includegraphics[trim=0cm 0cm 0cm 0cm, width=.36\textwidth]{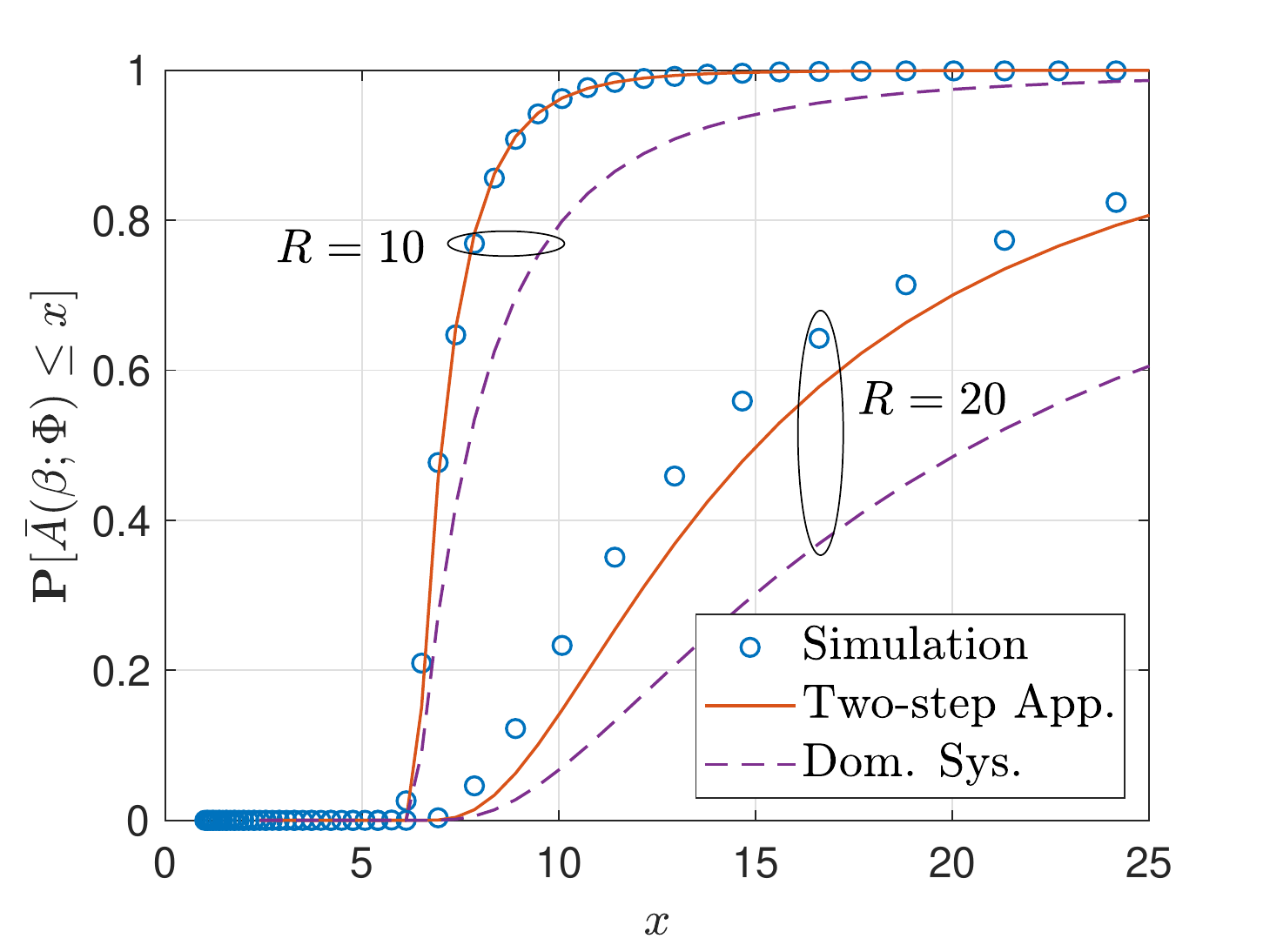}
\hspace{-7mm}\includegraphics[trim=0cm 0cm 0cm 0cm, width=.36\textwidth]{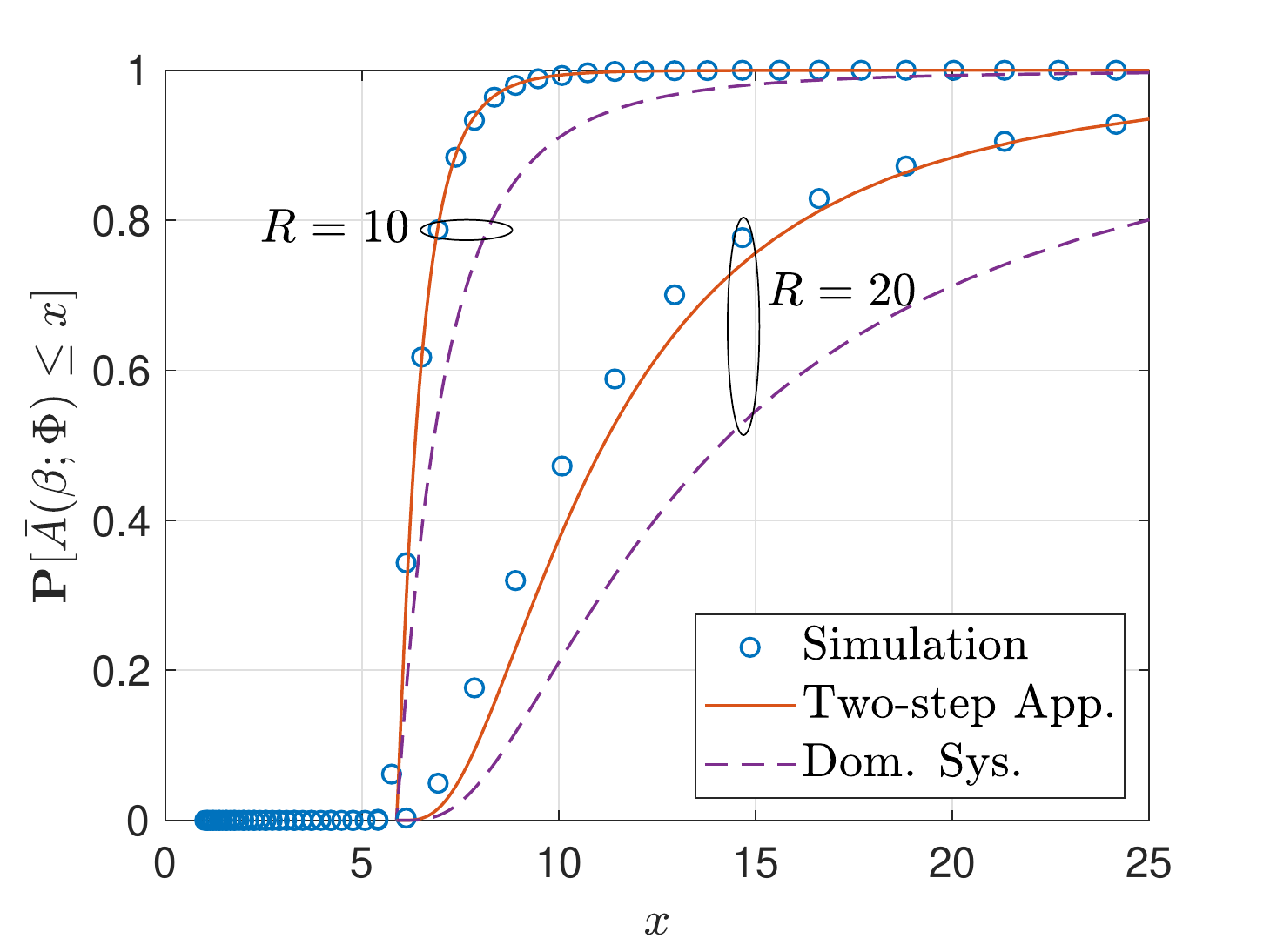}
 \caption{Verification of the proposed analytical framework:  Average of the conditional mean PAoI (Left) and the spatial distribution of the conditional mean PAoI  for non-preemptive (Middle) and preemptive (Right) disciplines.}
 \label{Fig:CDF_Mean_Peak_AoI}
\end{figure*}

 {Fig. \ref{Fig:CDF_Mean_Peak_AoI} verifies the proposed two-step approach based analytical framework using the simulation results for different values of link distances $R$.  Note that $R\in[10,20]$ is a wide enough range when $\lambda_{\rm sd}=10^{-3}$ links/m$^{2}$ (for which the dominant interfering source lies  at an average distance of around 15 m). Fig.  \ref{Fig:CDF_Mean_Peak_AoI} (Left) shows the spatial average of the conditional mean PAoI for both queuing disciplines, whereas Fig. \ref{Fig:CDF_Mean_Peak_AoI} (Middle) and (Right) show the spatial distribution of the conditional mean PAoI for the non-preemptive (obtained using Cor. \ref{cor:Distribution_AoI_I}) and preemptive  (obtained using Cor. \ref{cor:Distribution_AoI_II}) disciplines, respectively. From these figures it is quite apparent that the proposed two-step approach provide significantly tighter bounds as compared to the conventional dominant system based approaches. The spatial distribution (for larger $R$)  is somewhat loose as compared to the spatial average. This is because an additional approximation (beta approximation) is used to obtain the distribution of $\bar{A}(\beta;\Phi)$.    }

The performance trends of the  conditional mean PAoI $\bar{A}(\beta;\Phi)$ with respect to  the $\sir$ threshold $\beta$ and the path-loss exponent $\alpha$  are presented in Fig. \ref{Fig:Mean_SD_AoI_vs_Beta}. 
\begin{figure}[h]
\centering
 \includegraphics[trim=0cm 0cm 0cm 0cm, width=.5\textwidth]{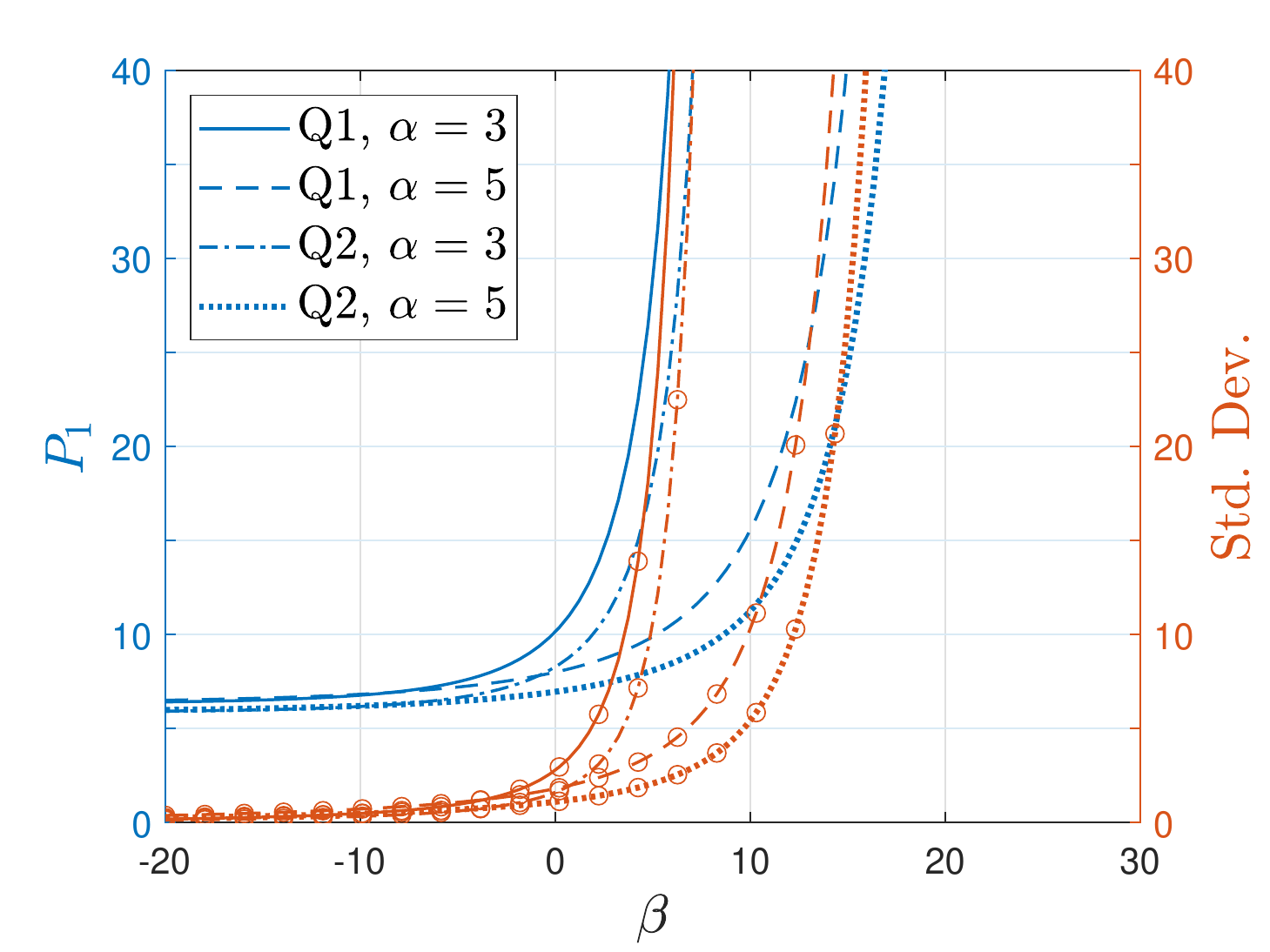}
 \caption{Mean and SD of $\bar{A}(\beta;\Phi)$ with respect to $\sir$ threshold $\beta$ for $R=10$. {The plane curves correspond to the mean, whereas the circular marked curves correspond to Std. Dev.}    }
 \label{Fig:Mean_SD_AoI_vs_Beta}
\end{figure}
It may be noted that the success probability decreases with the increase in $\beta$ and decrease in $\alpha$. As a result, we can observe from Fig. \ref{Fig:Mean_SD_AoI_vs_Beta} that the mean of $\bar{A}(\beta;\Phi)$,  i.e., $P_1$, also degrades with respect to these parameters in the same order (under both types of queues). As expected, it is also observed that $P_1$ increases sharply around the value of $\beta$ where the success probability {approaches  zero}.  On the other hand,  $P_1$ converges to a constant value as  $\beta$ {approaches zero} where the success probability is almost one. In this region, $P_1$ only depends on the packet arrival rate $\lambda_{\rm a}$ and medium access probability $\xi$. For $\lambda_{\rm a}=0.3$ and $\xi=0.5$,  we can obtain $P_1\approx 6.33$ for non-preemptive queue and $P_1\approx 5.56$ for preemptive queue  by plugging $\mu_\Phi=1$ in   \eqref{eq:Mean_Peak_Age_Given_Phi} and   \eqref{eq:Conditional_Mean_Peak_AoI_2}, respectively. These values of $P_1$ can be verified from Fig. \ref{Fig:Mean_SD_AoI_vs_Beta} when $\beta \leq 0$.

Further, it can be observed from Fig. \ref{Fig:Mean_SD_AoI_vs_Beta} that the standard deviations of $\bar{A}(\beta;\Phi)$ follow a similar trend as that of $P_1$ except at $\beta\to 0$. This implies that the second moment of  $\bar{A}(\beta;\Phi)$, i.e., $P_2$, increases at a faster rate than the one of $P_1^2$ with respect to $\beta$ and $\alpha$.   
In fact, this trend of $P_1$ and $P_2$ also holds for the other systems parameters which we discuss next. 

\begin{figure}[h]
\centering
 \includegraphics[trim=0cm 0cm 0cm 0cm, width=.5\textwidth]{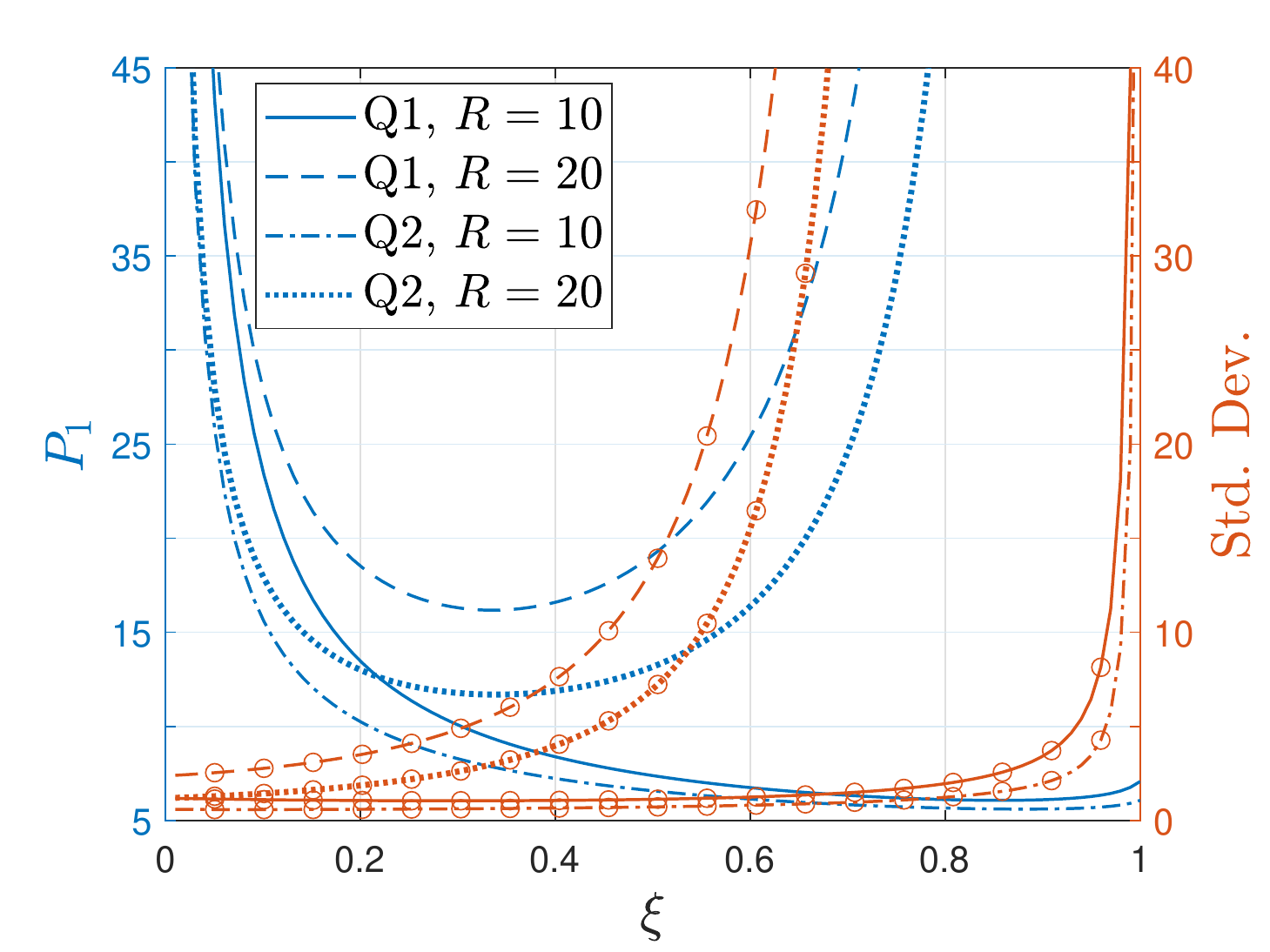}
 \hspace{-3mm} \includegraphics[trim=0cm 0cm 0cm 0cm, width=.5\textwidth]{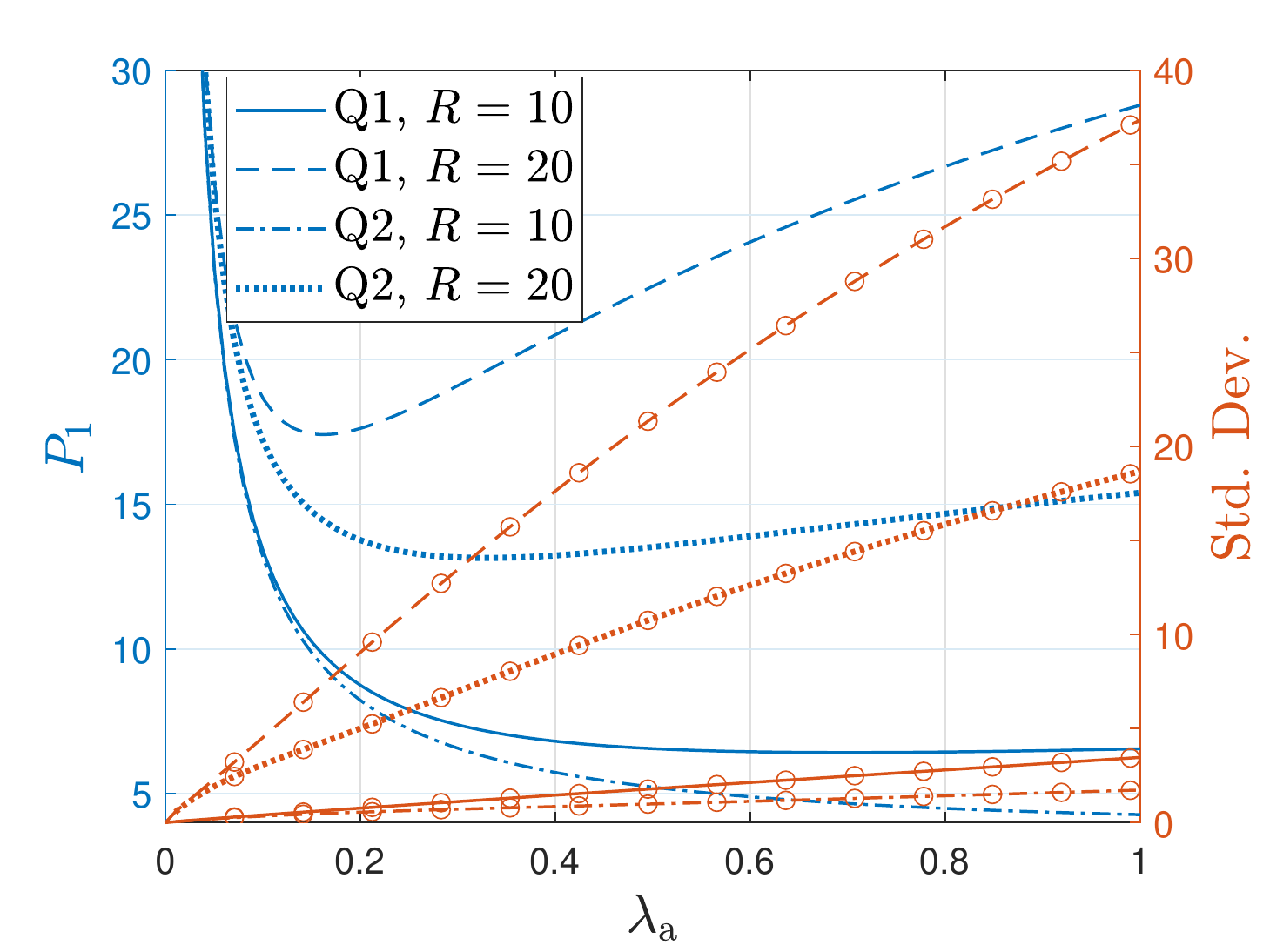}
 \caption{Mean and SD of $\bar{A}(\beta;\Phi)$ with respect to medium access probability $\xi$. {The plane curves correspond to the mean, whereas the circular marked curves correspond to Std. Dev.}   }
 \label{Fig:Mean_SD_AoI_vs_AccessProb_ArrivalRate}
\end{figure}
Now, we present the performance trends of  the mean and standard deviation of the conditional mean PAoI  with respect to the medium access probability $\xi$ and the status update arrival rate $\lambda_{\rm a}$ in Fig. \ref{Fig:Mean_SD_AoI_vs_AccessProb_ArrivalRate} (left) and Fig. \ref{Fig:Mean_SD_AoI_vs_AccessProb_ArrivalRate} (right), respectively. From these results, it can be seen that $P_1$ approaches to infinity as $\lambda_{\rm a}$ and/or $\xi$ drop to zero.   This is because the expected inter-arrival times  between updates approach infinity as $\lambda_{\rm a}\to 0$ and the expected delivery time approaches infinity as $\xi\to 0$.  
On the contrary, increasing $\lambda_{\rm a}$ and $\xi$  reduces the inter-arrival and delivery times of the status updates which  reduces $P_1$. However, $P_1$ again increases with further increase in $\lambda_{\rm a}$ and $\xi$. This is because the activities of the interfering SD pairs increase significantly at higher values of  $\lambda_{\rm a}$ and $\xi$, which causes severe interference and hence increases $P_1$. However, its rate of increase depends upon the success probability, which further depends upon the wireless link parameters such as the link distance $R$,  $\sir$ threshold $\beta$, path-loss exponent $\alpha$. For example, the figure shows that  $P_1$ increases at a faster rate with $\lambda_{\rm a}$ and $\xi$   when  $R$ is higher.  

{ Fig. \ref{Fig:Mean_AoI_vs_lam_sd} (Left) shows an interesting  fact that the same level of mean AoI performance (below a certain threshold) can be achieved across the different densities of SD pairs $\lambda_{\rm sd}$ just by controlling the medium access probability $\xi$.
The figure shows that    higher $\xi$ is a preferable choice for the lightly dense network (i.e., smaller $\lambda_{\rm sd}$). This is because of the tolerable interference in slightly dense scenario allows to choose higher $\xi$ without affecting success probability which in turn provides to the better service rate. On the hand, smaller $\xi$ is a preferable choice for highly  dense network (i.e., larger $\lambda_{\rm sd}$). This is because of the smaller $\xi$ can help to manage the severe interference in  dense scenario which in turn gives the better service rate.
\setcounter{figure}{9}
    \begin{figure}[h]
\centering
 \hspace{-5mm}\includegraphics[trim=0cm 0cm 0cm 0cm, width=.5\textwidth]{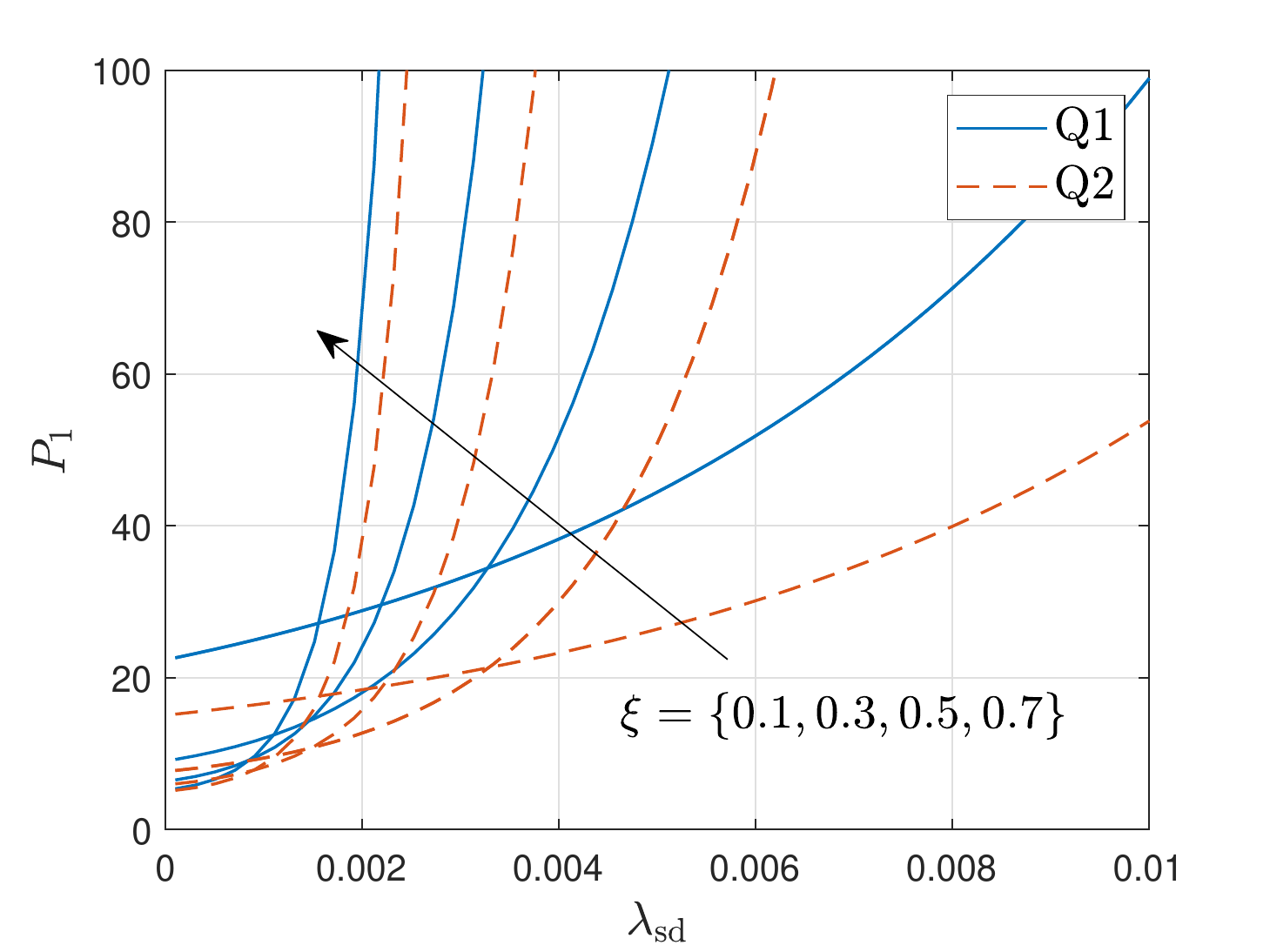}
 \hspace{-1mm} \includegraphics[trim=0cm 0cm 0cm 0cm, width=.5\textwidth]{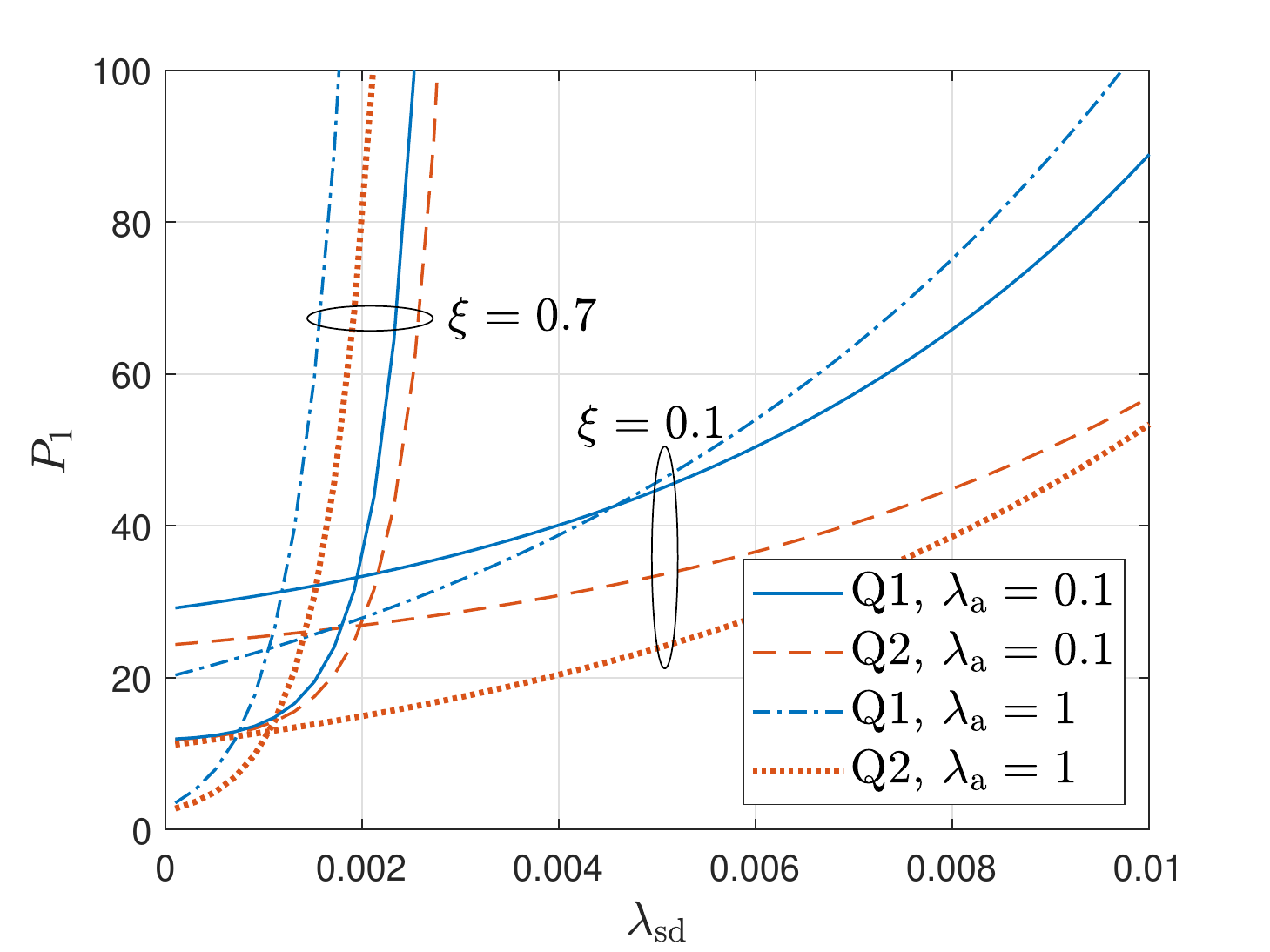}
 \caption{Mean of $\bar{A}(\beta;\Phi)$ with respect to the density of SD pairs $\lambda_{\rm a}$.}
 \label{Fig:Mean_AoI_vs_lam_sd}
\end{figure}
In addition,  Fig. \ref{Fig:Mean_AoI_vs_lam_sd} (Right) shows that $\xi$ allows to manage the mean AoI performance more efficiently as compare to the update  arrival rate $\lambda_{\rm a}$ for different network densities. Furthermore,  from the curves corresponding to $Q_1$ for $\xi=0.1$ in Fig. \ref{Fig:Mean_AoI_vs_lam_sd} (Right), it is interesting note that lowering $\lambda_{\rm a}$ with the increase of $\lambda_{\rm sd}$ allows to achieve better mean AoI performance. This follows due to the smaller $\lambda_{\rm a}$ aids to reduce the activities of sources which is beneficial to alleviate interference and thus provide better service rate in the dense scenario. }       

\section{Conclusion}
This paper considered a large-scale wireless network consisting of SD pairs whose locations follow a bipolar PPP. The source nodes are supposed to keep the information status at their corresponding destination nodes fresh by sending status updates over time. The AoI metric was used to measure the freshness of information at the destination nodes. For this system setup, we developed a novel stochastic geometry-based approach that allowed us to derive a tight upper bound on the spatial moments of the conditional mean PAoI {with no storage facility under preemptive and non-preemptive queuing disciplines.  The non-preemptive queue  drops the newly arriving updates  until the update in service is successfully delivered, whereas the preemptive queue discards the older update in service (or, retransmission) upon the arrival of a new update.}

Our analysis provides several useful design insights. For instance, the analytical results demonstrate the impact of the update arrival rate{, medium access probability and} wireless link parameters on the spatial mean and variance of the conditional mean PAoI.  The key observation was that preemptive queue reduces the mean PAoI almost by a factor two compared to the non-preemptive queue in the asymptotic regime of the update arrival rate when  the success probability is relatively smaller. 
{  Our numerical results also reveal that the medium access probability plays an important role as compared to the update arrival rate for ensuring better mean AoI performance under different network densities.}

{ As a promising avenue of future work, one can  analyze AoI for a system wherein the destinations are collecting status updates from multiple sources in their vicinity regarding a common physical random process.}
\appendix
\subsection{Proof of Lemma \ref{lemma:moments_SuccessProb}}
\label{app:moments_SuccessProb}
The $b$-th  moment of the conditional successful probability $\mu_\Phi$ given in \eqref{eq:conditional_SuccessProb} is
\begin{align*}
M_b&=\E_{\Phi,p_\x}\left[\prod_{\x\in\Phi}\left(1-\frac{p_\x}{1+\beta^{-1} R^{-\alpha}\|\x\|^{\alpha}}\right)^b\right]=\E_{\Phi}\left[\prod_{\x\in\Phi}\E_{p_\x}\left(1-\frac{p_\x}{1+\beta^{-1} R^{-\alpha}\|\x\|^{\alpha}}\right)^b\right],
\end{align*} 
where the last equality follows from the  assumption of independent activity $p_\x$ for $\forall \x\in\Phi$. Now, using the binomial expansion $(1-x)^b=\sum_{m=0}^\infty (-1)^m{b \choose m} x^m$ for a general  $b\in\Z$, we can write
\begin{align*}
M_b&=\E_{\Phi}\left[\prod_{\x\in\Phi}\sum_{m=0}^b(-1)^m{b \choose m}\frac{\bar{p}_m}{(1+\beta^{-1} R^{-\alpha}\|\x\|^{\alpha})^m}\right],
\end{align*} 
where $\bar{p}_m=\E[p_\x^m]$ is the $m$-th moment of the activity. Further, applying the probability generating functional ($\pgfl$) of homogeneous PPP $\Phi$, we obtain 
\begin{align*}
M_b&=\exp\left(-2\pi\lambda_\text{sd}\int_0^\infty\left[1-\sum_{m=0}^\infty (-1)^m{b \choose m}\frac{\bar{p}_m}{\left({1+\beta^{-1} R^{-\alpha}r^{\alpha}}\right)^m}\right]r{\rm d}r\right),\\
&=\exp\left(-2\pi\lambda_\text{sd}\sum_{m=1}^\infty (-1)^{m+1}{b \choose m}\bar{p}_m\int_0^\infty\frac{1}{(1+\beta^{-1} R^{-\alpha}r^{\alpha})^m} r{\rm d}r\right).\numberthis\label{eq:Mb_Appendix}
\end{align*}
Using \cite[Appendix A]{Martin2016Meta}, the inner integral can be evaluated as
\begin{align*}
\int_0^\infty\frac{(\beta R^{\alpha})^m}{\left(\beta R^{\alpha}+r^{\alpha}\right)^m} r{\rm d}r&=\frac{1}{\alpha}(\beta R^\alpha)^m\int_0^\infty\frac{t^{1-\delta}}{\left(\beta R^\alpha+t\right)^m} {\rm d}t=(-1)^{m+1}\frac{\beta^\delta R^2}{\alpha}{\delta -1 \choose m-1}\frac{\pi}{\sin(\delta\pi)}.
\end{align*}
Finally, substituting the above solution in \eqref{eq:Mb_Appendix}, we get \eqref{eq:Moments_mu_Phi}.
\subsection{Proof of Corollary \ref{cor:AoI_M12_Q1}}
\label{app:AoI_M12_Q1}
Solving \eqref{eq:Moment_Peak_AoI_a} for $b=\{1,2\}$ and then substituting $M_{-1}$ and $M_{-2}$ from Lemma \ref{lemma:moments_SuccessProb}, we obtain \eqref{eq:Mean_Peak_AoI}-\eqref{eq:Variance_Peak_AoI}. 
From the definition of $C_{\zeta_o}(b)$, we can directly determine 
 $$C_{\zeta_o}(-1)=\E\bigg[\sum_{m=1}^\infty {-1 \choose m}{\delta - 1 \choose m-1}\zeta_o^m\bigg]=-\E\left[\zeta_o(1-\zeta_o)^{\delta-1}\right].$$
 Now, for $b=-2$, let $C_{\zeta_o}(-2)=\E[B(\zeta_o,-2)]$ where 
 \begin{align*}
B(\zeta_o,-2)&=\sum_{m=1}^\infty {-2 \choose m}{\delta - 1 \choose m-1}\zeta_o^m,\\
&=\sum_{m=1}^\infty (-1)^{2m-1}(m+1){m - 1-\delta \choose m-1}\zeta_o^m,\\
&=-\sum_{m=1}^\infty (m-\delta)\frac{\Gamma(m-\delta)}{\Gamma(1-\delta)\Gamma(m)}\zeta_o^m + (\delta+1){m-1-\delta  \choose m-1}\zeta_o^m,\\
&=-\sum_{m=1}^\infty \frac{\Gamma(2-\delta)}{\Gamma(2-\delta)}\frac{\Gamma(m-\delta+1)}{\Gamma(1-\delta)\Gamma(m)}\zeta_o^m + (\delta+1){m-1-\delta  \choose m-1}\zeta_o^m, \\
&=-\sum_{m=1}^\infty \frac{\Gamma(2-\delta)}{\Gamma(1-\delta)}{m -\delta \choose m-1}\zeta_o^m + (\delta+1){m-1-\delta  \choose m-1}\zeta_o^m, \\
&=-(1-\delta)\sum_{m=1}^\infty {m -\delta \choose m-1}\zeta_o^m - (\delta+1)\sum_{m=1}^\infty{m-1-\delta  \choose m-1}\zeta_o^m, \\
&=(\delta-1)\zeta_o\sum_{l=0}^\infty {l+1 -\delta \choose l}\zeta_o^l - (\delta+1)\zeta_o\sum_{l=0}^\infty{l-\delta  \choose l}\zeta_o^l \\
&=(\delta-1)\zeta_o(1-\zeta_o)^{\delta-2} - (\delta+1)\zeta_o(1-\zeta_o)^{\delta-1}. 
 \end{align*}
 Finally, taking expectation of $B(\zeta_o,-2)$ with respect to $\zeta_o$ provides $C_{\zeta_o}(-2)$.
\subsection{Proof of Lemma \ref{lemma:Tk_pmf}}
\label{app:Tk_pmf}
Since the arrival process is Bernoulli, inter-arrival times $\tau_n$s between updates are i.i.d. and also follow geometric distribution with parameter $\lambda_\text{a}$.  From the fact that the sum of independent geometric random variables follows negative binomial distribution, we know 
\begin{equation*}
\P[X_n=K]= {K-1 \choose n-1}  \lambda_\text{a}^n(1-\lambda_\text{a})^{K-n}.
\end{equation*}
 Now, we derive the distribution of $\hat{T}_k$ using the above $\pmf$s of $T_k$ and $X_{N_{k}}$. We start by first deriving  the probabilities for the boundary values of $\hat{T}_k$ for given $T_k$.  Since at least one slot is required for the successful update delivery, we have  $1\leq \hat{T}_k\leq T_k$.
From Fig. \ref{Fig:Illustration_Tk}, it can be observed that { $\hat{T}_k=1$ either if $T_k=1$ or if there is a new arrival in the $(t_k+T_k)$-th time slot}. Therefore, we can write
 \begin{align*}
 \P[\hat{T}_k=1]&=\P[T_k=1] + \P[{\text{New Arrival Occurs in the}~(t_k+T_k)\text{-th slot}}]\P[T_k>1]\\
 &=\xi\mu_\Phi + \lambda_\text{a}(1-\xi\mu_\Phi).\numberthis\label{eq:pmf_Tk_Sk_1}
 \end{align*} 
 It is worth noting that $ \P[\hat{T}_k=1]$ is independent of $T_k$.
The other boundary case, i.e.,  $\hat{T}_k=T_k$, is possible only if there is no new update arrival in $[t_k, t_k^\prime)$. Thus,  we have 
 \begin{equation}
\P[\hat{T}_k=T_k]=(1-\lambda_\text{a})^{T_k-1}.
\label{eq:pmf_Tk_Sk_1a}
 \end{equation} 
 Now, we determine the probability of $\hat{T}_k=m$ for $m\in\{2,\dots,T_k-1\}$.
 For $\hat{T}_k=m$, the following  two conditions must hold so that we get $X_{N_k}=T_k-m$, 
 \begin{itemize}
 \item the number of new arrivals within $T_k$ are $N_k\in\{1,\dots,T_k-m\}$, and
  \item there  should not be a new arrival in $[t_k^\prime-m, t_k^\prime]$. 
\end{itemize}
 Therefore, for a given $T_k=s$,  we have
 \begin{align*}
  \P[\hat{T}_k=m|T_k=s]&=\P[\text{no arrival in ~}[t_k^\prime-m, t_k^\prime)]\P[N_k=n \text{~such that~} X_{n}=s-m],\\
&=(1-\lambda_\text{a})^m\sum_{n=1}^{s+1-m}\P[X_{n}=s-m],\\
 &=(1-\lambda_\text{a})^m\sum_{n=1}^{s+1-m}{s-m \choose n-1}  \lambda_\text{a}^n(1-\lambda_\text{a})^{s-m-n}.\numberthis \label{eq:pmf_Tk_Sk_2}
 \end{align*}
 Thus, for $T_k=s>1$, from \eqref{eq:pmf_Tk_Sk_1}-\eqref{eq:pmf_Tk_Sk_2},  we get
 \begin{equation}
 \P[\hat{T}_k=m|T_k=s]=\begin{cases}\lambda_\text{a},&\text{~for~}m=1,\\
\sum_{n=1}^{s+1-m}{s-m \choose n-1}  \lambda_\text{a}^n(1-\lambda_\text{a})^{s-n},&\text{~for~}m=2,\dots,s-1,\\
(1-\lambda_\text{a})^{m-1},&\text{~for~}m=s.
\end{cases}
\label{eq:pmf_Tk_Sk}
 \end{equation}
Finally, we can determine the $\pmf$ of $\hat{T}_k$ for $m>1$  as follows
 \begin{align*}
& \P[\hat{T}_k=m]= \sum_{s=m}^\infty \P[\hat{T}_k=m|T_k=s]\P[T_k=s],\\
 &=\P[\hat{T}_k=m|T_k=m]\P[T_k=m]+\sum_{s=m+1}^\infty \P[\hat{T}_k=m|T_k=s]\P[T_k=s],\\
 &=(1-\lambda_\text{a})^{m-1}\xi\mu_\Phi(1-\xi\mu_\Phi)^{m-1}+\underbrace{\sum_{s=m+1}^\infty\xi\mu_\Phi(1-\xi\mu_\Phi)^{s-1}\sum_{n=1}^{s+1-m}{s-m \choose n-1}  \lambda_\text{a}^n(1-\lambda_\text{a})^{s-n}}_{\ncalB}.
 \end{align*}
 where the last equality follows using \eqref{eq:Tk_pmf} and \eqref{eq:pmf_Tk_Sk}.
Substituting $s-m=l$ in the above term $\ncalB$, we get
\begin{align*}
\ncalB&=\sum_{l=1}^\infty\xi\mu_\Phi(1-\xi\mu_\Phi)^{l+m-1}\sum_{n=1}^{l+1}{l \choose n-1}  \lambda_\text{a}^n(1-\lambda_\text{a})^{l+m-n},\\
&=\xi\mu_\Phi(1-\xi\mu_\Phi)^{m} (1-\lambda_\text{a})^{m}
 \sum_{z=2}^\infty(1-\xi\mu_\Phi)^{z-2}\sum_{n=1}^{z}{z-1 \choose n-1}  \lambda_\text{a}^n(1-\lambda_\text{a})^{z-1-n}. \end{align*}
 Next, by substituting $$\sum_{n=1}^{l}{l-1 \choose n-1}  \lambda_\text{a}^n(1-\lambda_\text{a})^{l-n}=\lambda_\text{a},$$ 
and $z-2=a$,  we obtain
  \begin{align*}
\ncalB&=\lambda_\text{a}\xi\mu_\Phi(1-\xi\mu_\Phi)^{m} (1-\lambda_\text{a})^{m-1}
 \sum_{a=0}^\infty(1-\xi\mu_\Phi)^{a}
 =\lambda_\text{a}(1-\xi\mu_\Phi)^{m} (1-\lambda_\text{a})^{m-1},
\end{align*}
 where the last equality follows using the power series $\sum_{n=0}^\infty x^n=\frac{1}{1-x}$ for $|x|<1$. Finally,  substituting $\ncalB$ in $\P[\hat{T}_k=m]$,  we get
 \begin{align*}
 \P[\hat{T}_k=m]&=\xi \mu_\Phi (1-\xi\mu_\Phi)^{m-1}(1-\lambda_\text{a})^{m-1}+ \lambda_\text{a}(1-\xi\mu_\Phi)^{m}(1-\lambda_\text{a})^{m-1},\\
 &=\left(\xi \mu_\Phi + \lambda_\text{a}(1-\xi\mu_\Phi)\right)(1-\xi\mu_\Phi)^{m-1}(1-\lambda_\text{a})^{m-1},\numberthis \label{eq:pmf_Tk_2}
 \end{align*}
for $m> 1$.
 Therefore, using \eqref{eq:pmf_Tk_Sk_1} and \eqref{eq:pmf_Tk_2}, we obtain \eqref{eq:pmf_Tk}.
 \subsection{Poof of Theorem \ref{thm:AoI_Mb_Q2}}
 \label{app:AoI_Mb_Q2}
 Let ${\cal S}_{\rm a}=\xi\mu_\Phi (1-\lambda_{\rm a})+\lambda_{\rm a}$ be the denominator of the last term in \eqref{eq:Conditional_Mean_Peak_AoI_2}. 
Since $\lambda_{\rm a}\in (0,1)$, ${\cal S}_{\rm a}$ represents the convex combination of $1$ and $\xi\mu_\Phi$. As $\xi\mu_\Phi\in(0,1)$,  we have $0< {\cal S}_{\rm a}<1$. Therefore, using the following binomial expansion  $$(1+x)^{-n}=\sum_{k=0}^\infty (-1)^k{n+k-1\choose k} x^k \text{~for~} |x|<1, $$ we can write 
\begin{align*}
{\cal S}_{\rm a}^{-n}=\left(1+({\cal S}_{\rm a}-1)\right)^{-n}=\sum_{k=0}^\infty (-1)^k{n+k-1 \choose k} ({\cal S}_{\rm a}-1)^k.
\end{align*}
Note that
$({\cal S}_{\rm a}-1)^k=\left[\xi\mu_\Phi(1-\lambda_{\rm a}) - (1-\lambda_{\rm a})\right]^k=(-1)^k(1-\lambda_{\rm a})^k(1-\xi\mu_\Phi)^k.$ 
Using this, we can obtain the expectation of $\mu_\Phi^{-m}{\cal S}_{\rm a}^{-n}$ as 
\begin{align*}
S(n;m)&=\E\left[\mu_\Phi^{-m}{\cal S}_{\rm a}^{-n}\right],\\
&=\E\bigg[\mu_\Phi^{-m}\sum_{k=0}^\infty {n+k-1 \choose k} (1-\lambda_{\rm a})^k (1-\xi\mu_\Phi)^k\bigg],\\
&=\sum_{k=0}^\infty {n+k-1 \choose k} (1-\lambda_{\rm a})^k\sum_{l=0}^k (-1)^l{k \choose l} \xi^lM_{l-m}.
\end{align*} 
 Using this and \eqref{eq:Conditional_Mean_Peak_AoI_2}, we now derive the $b$-th moment of $\bar{A}_{\rm P}(\beta;\Phi)$ as
 \begin{align*}
 P_b&=\E\bigg[\bigg(\ncalZ_{\text{a}}+\frac{1}{\xi\mu_\Phi} + \frac{1}{\xi\mu_\Phi (1-\lambda_\text{a})+ \lambda_\text{a}}\bigg)^b\bigg]=\sum_{l+m+n=b}{b \choose l,m,n}\ncalZ_{\text{a}}^l\xi^{-m}\E\left[\mu_\Phi^{-m}{\cal S}_{\rm a}^{-n}\right].
 \end{align*}
Further, substituting  $\E\left[\mu_\Phi^{-m}{\cal S}_{\rm a}^{-n}\right]=S(n;m)$ in the above expression, we obtain \eqref{eq:TypeII_Qb}.
Since the moments of the upper bound  of the activity probabilities of interfering sources  (given in \eqref{eq:Activity_Moments_DomSys})  are used for evaluating the moments  of $\mu_\Phi$, the resulting $b$-th moment of $\bar{A}_{\rm P}(\beta;\Phi)$  in \eqref{eq:TypeII_Qb}  is in fact an upper bound.
\subsection{Proof of Corollary \ref{cor:AoI_M12_Q2}}
\label{app:AoI_M12_Q2}
 Using  $\E\left[\mu_\Phi^{-m}{\cal S}_{\rm a}^{-n}\right]=S(n;m)$ and \eqref{eq:Conditional_Mean_Peak_AoI_2},  the mean of $\bar{A}_{\rm P}(\beta;\Phi)$ can be determined as
\begin{align*}
P_1&={\cal Z}_\text{a}+\xi^{-1}\E_\Phi\left[\mu_\Phi^{-1}\right]+\E_\Phi\left[{\cal S}_{\rm a}^{-1}\right]={\cal Z}_\text{a}+\xi^{-1}M_{-1}+S(1;0).
\end{align*} 
Similarly, the second moment of $\bar{A}_{\rm P}(\beta;\Phi)$ can be determined as
\begin{align*}
P_2&=\E\left[\left(\ncalZ_{\text{a}}+({\xi\mu_\Phi})^{-1} + {\cal S}_{\rm a}^{-1}\right)^2\right],\\
&=\ncalZ_{\text{a}}^2 + 2\ncalZ_{\text{a}}\xi^{-1}M_{-1}+\xi^{-2}M_{-2} + 2{\cal Z}_{\rm a}  S(1;0)+ 2\xi^{-1}S(1;1)+S(2;0).
\end{align*}
Further, substitution of $M_{l}$  from Lemma \ref{lemma:moments_SuccessProb} provides the first two moments of  $\bar{A}_{\rm P}(\beta;\mu_\Phi)$ as in \eqref{eq:Mean_AoI_2} and \eqref{eq:M2_AoI_2}. 
The values of $C_{\zeta_o}(l)$ for $l\in\{-1,-2\}$ directly follow from  Corollary \ref{cor:AoI_M12_Q1}. However, the upper limit of  the summation in $C_{\zeta_o}(l)$ for $l>0$ reduces to $l$  (refer to Appendix \ref{app:moments_SuccessProb}).
Since the lower bounds of the moments of $\mu_\Phi$ obtained from the  two-step analysis are used here, the moments of  $\bar{A}_{\rm P}(\beta;\Phi)$ given in \eqref{eq:Mean_AoI_2} and \eqref{eq:M2_AoI_2} are the upper bounds (please refer to Section \ref{subsec:TwoStepAnalysis} for more details about the constructions and assumptions).



\end{document}